\documentclass{amsart}

\usepackage{graphicx}

\newtheorem{theorem}{Theorem}[section]
\newtheorem{lemma}[theorem]{Lemma}
\newtheorem{proposition}[theorem]{Proposition}

\theoremstyle{definition}
\newtheorem{definition}[theorem]{Definition}
\newtheorem{example}[theorem]{Example}
\newtheorem{corollary}[theorem]{Corollary}

\theoremstyle{remark}
\newtheorem{remark}[theorem]{Remark}

\numberwithin{equation}{section}

\usepackage{ulem}
\usepackage{amsrefs}

\usepackage{color}

\let\oldmarginpar\marginpar
\renewcommand\marginpar[1]{\-\oldmarginpar[\raggedleft\footnotesize #1]
{\raggedright\footnotesize #1}}

\usepackage[T1]{fontenc}
\usepackage{lmodern}
\usepackage[latin1]{inputenc}

\usepackage{setspace}
\usepackage{indentfirst}

\usepackage{multicol}
\usepackage{geometry}
\usepackage{hyperref}

\usepackage{amsmath,amssymb}
\usepackage{mathrsfs}

\usepackage{amsthm}

\theoremstyle{plain}

\newcommand{\cV}{\mathcal{V}}

\newcommand{\cE}{\mathcal{E}}

\newcommand{\C}{\mathbb{C}}
\newcommand{\CP}{\mathbb{C}\mathbf{P}}
\newcommand{\R}{\mathbb{R}}
\newcommand{\Z}{\mathbb{Z}}
\newcommand{\g}{\mathfrak{g}}
\newcommand{\F}{\mathcal{F}}
\newcommand{\A}{\mathcal{A}}

\DeclareMathOperator{\id}{id}
\DeclareMathOperator{\Ima}{Im}

\hypersetup{bookmarksopen=true,pdfpagelayout=SinglePage,colorlinks=false,pdfstartview=Fit,pdfpagemode=FullScreen}

\begin{document}

\title{Poisson Cohomology of holomorphic toric Poisson manifolds. I.}

\author{Wei Hong}
\address{School of Mathematics and Statistics, Wuhan University, China}
\address{Mathematics Research Unit, University of Luxembourg, Luxembourg}
\email{hong\textunderscore  w@whu.edu.cn}


\begin{abstract}
A holomorphic toric Poisson manifold is a nonsingular toric variety equipped with a holomorphic Poisson structure, which is invariant under the torus action. 
In this paper, we computed the Poisson cohomology groups for all holomorphic toric Poisson structures on $\CP^n$, with the stand Poisson structure on $\CP^n$ 
as a special case. 
We also computed the algebraic and the formal Poisson cohomology groups of holomorphic toric Poisson structures on $\C^n$. 
\end{abstract}

\maketitle


\section{introduction}
Holomorphic Poisson manifolds play an important role in modern mathematics.
Many interesting works appeared in recently years.
The algebraic geometry of the Poisson brackets on projective spaces were studied by Bondal \cite{ Bondal} and Polishchuk \cite{Polishchuk}. 
 The close relation of holomorphic Poisson structures with generalized complex 
 geometry and mathematical physics were revealed by 
 Hitchin \cite{Hitchin 06, Hitchin 11} and Gualtieri \cite{Gualtieri 11}.  Deformations of holomorphic Poisson structures appeared in the 
 work of \cite{Goto 10}, \cite{Hitchin 12} and \cite{Kim 14}. 
In \cite{B-G-Y 06, G-Y 09}, Brown, Goodear and Yakimov  studied the standard 
Poisson structures on affine spaces and flag varieties. 
Laurent-Gengoux, Sti\'{e}non and Xu \cite{L-S-X 08} described the Poisson cohomology 
of holomorphic Poisson manifolds using Lie algebroids. 
In various situations, the Poisson cohomology of holomorphic Poisson manifolds were computed \cite{Hong-Xu 11, Mayansky 15, C-F-P 16,  Poon 16}.

This paper is devoted to the study of the Poisson geometry of toric varieties, 
especially, the Poisson cohomology of holomorphic toric Poisson manifolds.
A holomorphic toric Poisson manifold is a nonsingular toric variety $X$, equipped with a holomorphic Poisson structure $\pi$, which is invariant under the torus action
( Notice that real toric Poisson structures were studied in \cite{Caine}).
Holomorphic toric Poisson manifold is a special case of the ``$T$-Poisson manifold'' 
in the sense of \cite{Lu-Mouquin 15}. 
In this paper, we computed the Poisson cohomology groups for all holomorphic toric Poisson structures on $X=\CP^n$.
We also computed the algebraic Poisson cohomology groups and the formal Poisson cohomology groups for all holomorphic toric Poisson structures on $X=\C^n$.

This paper is organized in the following way. 
In Section \ref{section2}, we give some background knowledge. In Section \ref{section3} and Section \ref{section4}, we study the Poisson cohomology groups of $X=\CP^n$. In Section \ref{section5}, we study the Poisson cohomology groups of $X=\C^n$. 

The main results of this paper are:
\begin{itemize}
\item We give a weight spaces decomposition according to the $(\C^*)^n$-representation on the space of holomorphic vector fields and multi-vector fields on $X=\CP^n$ (Theorem \ref{CPn-multiVec-thm} and Proposition \ref{CPn-multiVec-prop}).
\item  We described the Poisson cohomology groups of all holomorphic toric Poisson structures on $X=\CP^n$  in Theorem \ref{CPn-cohomology-thm} and Proposition \ref{CPn-cohomology-prop}, which gives the weight spaces decomposition for the $(\C^*)^n$-representation on $H^{\bullet}_\pi(X)$.
\item  We computed the Poisson cohomology groups for the standard Poisson structure on $X=\CP^n$ in certain situations. ( Theorem \ref{ST-thm}, Proposition \ref{CP2-coh-prop} and Proposition \ref{CP3-coh-prop})
\item For all holomorphic toric Poisson structures on $X=\C^n$, we described the corresponding algebraic Poisson cohomology groups and the formal Poisson cohomology groups in Theorem \ref{Cn-cohomology-thm}, 
which developed Monnier's work in \cite{Monnier 02}.
\end{itemize}

{\bf Acknowledgements}\quad
I would like to thank Sam Evens,  Zhangju Liu, Jianghua Lu, Yannick Voglaire and Ping Xu for their helpful discussions and comments. 
Special thanks go to Yat Sun Poon and Bing Zhang for their valuable opinions to the draft of this paper.
I wish to express my deep gratitude to Martin Schlichenmaier for his support during the author's stay in Luxembourg. Hong's research was partially supported by NSFC grant 11401441 and FNR grant 5650104.

\section{Preliminary}\label{section2}
\subsection{Poisson cohomology of holomorphic Poisson manifolds}
\begin{definition}
A holomorphic Poisson manifold is a complex manifold $X$ equipped with a holomorphic bivector field $\pi$ such that $[\pi,\pi]=0$, where $[\cdot,\cdot]$ is the Schouten bracket.
\end{definition}
The Poisson cohomology of a holomorphic Poisson manifold is defined in the following way:
\begin{definition}
Let $(X, \pi)$ be a holomorphic Poisson manifold of dimension n.
 The Poisson cohomology  $H^\bullet_\pi (X) $ is
 the cohomology group of the complex of sheaves:
\begin{equation}
\mathcal{O}_{X}\xrightarrow{d_{\pi}}T_{X}\xrightarrow{d_{\pi}}.....
\xrightarrow{d_{\pi}}\wedge^{i-1}T_{X}\xrightarrow{d_{\pi}}\wedge^{i}T_{X}
\xrightarrow{d_{\pi}}\wedge^{i+1}T_{X}\xrightarrow{d_{\pi}}......
\xrightarrow{d_{\pi}}\wedge^{n}T_{X},
\end{equation}
where $d_{\pi}=[\pi,\cdot]$.
\end{definition}

\begin{lemma}\cite{L-S-X 08}
  The Poisson cohomology of a holomorphic
  Poisson manifold $(X,\pi)$ is isomorphic to the  total cohomology of the double complex \\
$$\begin{array}{ccccccc}
......& &......& &......& & \\
d_{\pi}\big\uparrow & & d_{\pi}\big\uparrow & & d_{\pi}\big\uparrow & &  \\
\Omega^{0,0}(X,T^{2,0}X) & \xrightarrow{\bar{\partial}} &
 \Omega^{0,1}(X, T^{2,0}X) & \xrightarrow{\bar{\partial}} &
 \Omega^{0,2}(X, T^{2,0}X) &\xrightarrow{\bar{\partial}}   &
 ......\\
 d_{\pi}\big\uparrow & & d_{\pi}\big\uparrow & & d_{\pi}\big\uparrow & &  \\
\Omega^{0,0}(X,T^{1,0}X) & \xrightarrow{\bar{\partial}} &
 \Omega^{0,1}(X, T^{1,0}X) & \xrightarrow{\bar{\partial}} &
 \Omega^{0,2}(X,T^{1,0}X) &\xrightarrow{\bar{\partial}}   &
 ......\\
 d_{\pi}\big\uparrow & & d_{\pi}\big\uparrow & & d_{\pi}\big\uparrow & &  \\
\Omega^{0,0}(X, T^{0,0}X) & \xrightarrow{\bar{\partial}} &
\Omega^{0,1}(X, T^{0,0}X) & \xrightarrow{\bar{\partial}} & \Omega^{0,2}(X, T^{0,0}X) & \xrightarrow{\bar{\partial}} &
 ......\\
\end{array}$$
\end{lemma}

\begin{lemma}\label{LSX-lem}
Let $(X,\pi)$ be a holomorphic Poisson manifold. If all the higher
cohomology groups $H^{i}(X,\wedge^{j}T_{X})$ vanish for $i>0$, then the
Poisson cohomology  $H^\bullet_\pi (X)$
is isomorphic to the cohomology of the complex
\begin{equation}
H^{0}(X,\mathcal{O}_{X})\xrightarrow {d_{\pi}}
H^{0}(X,T_{X})\xrightarrow{d_{\pi}}
H^{0}(X,\wedge^{2}T_{X})\xrightarrow{d_{\pi}}
\ldots \xrightarrow{d_{\pi}}H^{0}(X,\wedge^{n}T_{X}),
\end{equation}
where $d_{\pi}=[\pi,\cdot]$.
\end{lemma}

\subsection{Toric varieties}
In this section, we recall some classical knowledge of toric varieties.
One may consult \cite{Cox}, \cite{Fulton} and \cite{Oda}.
\begin{definition}
A toric variety is an irreducible variety $X$ such that
\begin{enumerate}
\item $(\C^*)^n$ is a Zariski open set of $X$, and
\item the action of $(\C^*)^n$ on itself extends to an action of $(\C^*)^n$ on $X$.
\end{enumerate}
\end{definition}

Toric varieties can also be described by a Lattice $N\cong\Z^n$ and a fan $\Delta$ in $N_{\mathbb{R}}=N\otimes_\Z\mathbb{R}\cong\mathbb{R}^n$. 

Let $M=Hom_\Z(N,\Z)$ and $T_N=Hom_\Z(M,\C^*)=N\otimes_{\Z}\C^*$. Then $T_N\cong(\C^*)^n$.
Moreover, we have $M\cong Hom(T_N,\C^*)$ and $N\cong Hom(\C^*, T_N)$.

Each element $m$ in $M$ gives rise to a character $\chi^m\in Hom(T_N,\C^*)$,
given by $$\chi^m(t)=\langle t,m\rangle \quad\text{for}\quad t\in T_N.$$
Each element $a$ in $N$ gives rise to a one-parameter subgroup 
$\gamma_a\in Hom(\C^*,T_N)$ given by
$$\gamma_a(\lambda)(m)=\lambda^{\langle a, m\rangle}\quad \text{for} \quad \lambda\in\C^*\quad \text{and}\quad m\in M.$$

Fix a $\Z$-basis $\{e_1,...e_n\}$ of $N$ and let $\{e_1^*,...e_n^*\}$ 
be the dual basis of $M$. Let $t_i=\langle t, e_i^*\rangle$ $(1\leq i\leq n)$.
Then there is an isomorphism 
$$T_N\cong(\C^*)^n: t\longleftrightarrow (t_1, t_2,...t_n),$$
where $t_1, t_2,...t_n\in\C^*$ are considered as the coordinates on $T_N$.

For $m=\sum_{i=1}^n m_i e_i^*$, we have $\chi^m=t_1^{m_1}t_2^{m_2}...t_n^{m_n}$,
which is a Laurent monomial on $T_N$. For $a=\sum_{i=1}^n a_i e_i$, 
the one-parameter subgroup $\gamma_a$ can be written as 
$\gamma_a(\lambda)=(\lambda^{a_1},...\lambda^{a_n})$ for $\lambda\in\C^*$.
 
 \begin{definition}
 A subset $\sigma$ of $N_\R$ is called a rational polyhedral cone (with apex at the origin $O$), if there there exist a finite number of elements $e_1, e_2,...,e_s$ in $N$ such
 that 
 \begin{align*}
 \sigma&=\R_{\geq 0}e_1+...\R_{\geq 0} e_s\\
 &=\{a_1 e_1+...+a_s e_s\mid a_i\in\R, a_i\geq 0~ \text{for all} ~0\leq i\leq s\},
 \end{align*}
 where we denote by $R_{\geq 0}$ the set of nonnegative real numbers.
 \begin{enumerate}
\item $\sigma$ is strongly convex if $\sigma\cap(-\sigma)={O}$.
\item The dimension of $\sigma$ is the dimension of the smallest subspace of $N_\R$ containing $\sigma$.
\end{enumerate}
 \end{definition}
 In this paper, a cone is always a rational polyhedral cone.

Let $M_\R=M\otimes_Z\R$. The canonical $\Z$-bilinear pairing 
$$\langle,\rangle: M\times N\rightarrow\Z$$
extends to a $\R$-bilinear pairing 
$\langle ,\rangle: M_\R\times N_\R\rightarrow\R$.
Given a cone $\sigma\in N_\R$, its dual cone in $M_\R$ is defined to be
 $$\sigma^{\vee}=\{x\in M_\R\mid \langle x,y\rangle\geq0~\text{ for all}~y\in\sigma\}.$$
 A face of $\sigma$ is a subset of $\sigma$, which can be written as 
$m^{\perp}\cap\sigma=\{x\in\sigma\mid \langle x,m\rangle=0\}$ for $m\in\sigma^{\vee}$.

\begin{definition}\label{defn-fan}
A fan in $N$ is a nonempty collection $\Delta$ of strongly convex rational polyhedral cones in $N_\R$ satisfying the following conditions:
\begin{enumerate}
\item Every face of any $\sigma\in\Delta$ is contained in $\Delta$.
\item For any $\sigma,\sigma' \in\Delta$, the intersection $\sigma\cap\sigma'$ is a face of both $\sigma$ and $\sigma'$.
\end{enumerate}
The union $|\Delta|=\cup_{\sigma\in\Delta}\sigma$ is called the support of $\Delta$.
 \end{definition}
 In this paper, we assume that fans are finite, i.e., a fan consist of only finite number of cones.
 
 Given a fan $\Delta$, the set of $k$-dimensional cones in $\Delta$ is denoted by 
 $\Delta(k)$ $(0\leq k\leq n)$. The primitive element of $\alpha\in\Delta(1)$ is the unique generator of $\alpha\cap N$, denoted by $e(\alpha)$.
 
 Let $S_{\sigma}=\sigma^{\vee}\cap M$. For a strongly convex rational polyhedral cone 
 $\sigma$ in $N_\R$, the semigroup algebra 
  $$\C[S_{\sigma}]=\bigoplus_{m\in S_{\sigma}}\C\chi^m$$ 
  is a finitely generated commutative $\C$-algebra. 
 The affine variety $U_{\sigma}=Spec(\C[S_{\sigma}])$ is a
  $n$-dimensional toric variety.
 
\begin{theorem}
Given a lattice $N\cong\Z^n$ and a fan $\Delta$ in $N_\R\cong R^n$, 
there exists a toric variety $X_\Delta$, obtained from the affine variety 
$U_\sigma, \sigma\in\Delta$, by gluing together 
$U_\sigma$ and $U_\tau$ along their common open subset $U_{\sigma\cap\tau}$ for all
$\sigma,\tau\in\Delta$. 
\end{theorem}

A cone $\sigma$ is called nonsingular if $\sigma$ can be written as 
$$\sigma=\R_{\geq 0}e_1+...\R_{\geq 0} e_s,$$ 
where $\{e_1, e_2,...,e_s\}$ is a subset of a $\Z$-basis of $N$.
\begin{theorem}
 Let $X_\Delta$ be the toric variety  associated with a fan $\Delta$ in $N_\R$. Then
 \begin{enumerate}
\item $X_\Delta$ is compact $\Longleftrightarrow$ $|\Delta|=N_\R$.
\item $X_\Delta$ is nonsingular $\Longleftrightarrow$ each $\sigma\in\Delta$ is nonsingular.
\end{enumerate}
 \end{theorem}

For a nonsingular toric variety $X_\Delta$, the action map 
$T_N\times X_\Delta\rightarrow X_\Delta$ is a holomorphic map. 
Let $N_\C=N\otimes_{\Z}\C$. Then $Lie(T_N)\cong N_\C$. 
By identification of $Lie(T_N)$ with $N_\C$,
the infinitesimal action of Lie algebra $Lie(T_N)$ on $X_\Delta$ defines a map
\begin{equation} \label{rho-def-eqn}
\rho:N_\C=N\otimes_\Z\C\rightarrow \mathfrak{X}(X_\Delta).
\end{equation}
The image of $\rho$ are holomorphic vector fields on $X_{\Delta}$.
For any $e\in N_\C$ and $m\in M$, we have
\begin{equation} \label{action-chi-eqn}
\rho(e)(\chi^m)=\langle e, m\rangle \chi^m,
\end{equation}
where $\chi^m$ is considered as a rational function on $X_\Delta$.
By abuse of notation, we denote the induced map
\begin{equation} \label{rho-def-eqn2}
\wedge^k N_\C\rightarrow\mathfrak{X}^k(X_\Delta)
\end{equation}
 also by $\rho$.
 
\begin{example}\label{Cn-exa1}
Let $X=\C^n$ and $(z_1,\ldots, z_n)$ be the standard coordinates on it. 
There is a nature embedding $(\C^*)^n\hookrightarrow\C^n$.
The $(\C^*)^n$-action on $X=\C^n$ defined by 
$$(t_1,t_2,\ldots,t_n)\cdot(z_1,z_2,\ldots,z_n)=(t_1 z_1, t_2 z_2,\ldots, t_n z_n)$$
makes $X=\C^n$ a toric variety.

Let $\{e_1=(1,0,,..,0),~..., ~e_n=(0,...,0,1)\}$ be the standard $\Z$-basis of $N=\Z^n$. 
Let $$\sigma=\sum_{i=1}^{n}\R_{\geq0}e_i.$$ 
Let the fan $\Delta$ be the collection of the cones of the following form:
$$\sum_{s=1}^{k}\R_{\geq0}e_{i_s},\quad \{i_1,i_2,...,i_k\}\subseteq\{1,...,n\}.$$
Then we have $X_\Delta=U_\sigma\cong\C^n$.

Let $\{e_1^*,e_2^*,...,e_n^*\}$ be the dual basis of $\{e_1,e_2,...,e_n\}$ in $M$.
Then we have $$\chi^{e_i^*}=z_i, \quad 1\leq i\leq n.$$
For $m=\sum_{i=1}^n m_i e_i^*$, the rational function $\chi^m$ on 
$X_\Delta=\C^n$ can be written as 
\begin{equation}\label{Cn-chi-eqn}
\chi^m=z_1^{m_1}...z_n^{m_n}.
\end{equation} 

\end{example}

\begin{example}\label{CPn-exa1}
Let $X=\CP^n$ and $[z_0, z_1,\ldots, z_n]$ be homogenous coordinates on it. 
The map $$(\C^*)^n\rightarrow\CP^n$$
defined by $(t_1,t_2,...,t_n)\mapsto[1, t_1,t_2,...,t_n]$ allows us to identify $(\C^*)^n$ with the Zariski open subset $\{[z_0,z_1,\ldots,z_n]\in\CP^n\mid z_i\neq 0, ~ 0\leq i\leq n\}$ of $\CP^n$.
The $(\C^*)^n$ action on $\CP^n$ given by 
$$(t_1,\ldots,t_n).[z_0,z_1,\ldots,z_n]=[z_0, t_1z_1,\ldots, t_n z_n]$$
makes $X=\CP^n$ to be a toric variety.  

Let $e_0=(-1,-1,...,-1), e_1=(1,0,,..,0),...,e_n=(0,...,0,1)$ be vectors in $N=\Z^n$. 
Let the fan $\Delta$ be the collection of the cones of the following form:
$$\sigma=\sum_{s=1}^{k}\R_{\geq0}e_{i_s},\quad \{i_1,i_2,...,i_k\}\subsetneq\{0,1,...,n\}.$$
Then we have $X_{\Delta}\cong\CP^n$.
Let
 $$\sigma_i=\sum_{s=1}^{n}\R_{\geq0}e_{i_s},\quad \{i_1,i_2,...,i_n\}=\{0,1,...,n\}\backslash\{i\}.$$
Then $U_{\sigma_i}$ can be identified with the affine open set 
$U_i=\{[z_0,z_1,...,z_n]\in\CP^n\mid z_i\neq 0\}$.

Choose the $\Z$-basis $\{e_1,e_2,...,e_n\}$ of $N$ and its dual basis
$\{e_1^*,e_2^*,...,e_n^*\}$ of $M$. 
Then $t_i=\chi^{e_i^*}$ $(1\leq i\leq n)$ can be considered as the affine coordinates on $U_0$, i.e.,  $t_i=\frac{z_i}{z_0}$.
And for $m=\sum_{i=1}^n m_i e_i^*$, the rational function $\chi^m$ on 
$X_\Delta=\CP^n$ can be written as 
\begin{equation}\label{CPn-chi-eqn}
\chi^m=t_1^{m_1}t_2^{m_2}...t_n^{m_n}=z_0^{m_0}z_1^{m_1}...z_n^{m_n},
\end{equation} 
where $m_0=-\sum_{i=1}^n m_i$. 
\end{example}

\subsection{Holomorphic toric Poisson structures}
\begin{definition}
Let $X$ be a nonsingular toric variety. If a holomorphic Poisson structure $\pi$ on $X$ is invariant under the torus action, then $\pi$ is called a holomorphic toric Poisson structure on $X$, and $X$ is called a holomorphic toric Poisson manifold.
\end{definition}

\begin{proposition}\label{ToricPoissonPro}
Let $X_\Delta$ be a nonsingular toric variety associated with a fan $\Delta$ in $N_\R$. 
Then the set of holomorphic toric Poisson structures on $X$ coincide with $\rho(\wedge^2 N_\C)$, where $\rho:\wedge^2 N_\C\rightarrow \mathfrak{X}^2(X_\Delta)$ 
is defined in \eqref{rho-def-eqn2}.
\end{proposition}

Proposition \ref{ToricPoissonPro} can be state in an equivalent way:

\begin{proposition}
Let $X_\Delta$ be a nonsingular toric variety associated with a fan $\Delta$ in $N_\R$. 
Suppose that $\{e_1, e_2,\ldots, e_n\}$ is a basis of $N\subset N_\C$. 
Let $v_i=\rho(e_i)$ $(1\leq i\leq n)$  be holomorphic vector fields on $X_\Delta$, 
where $\rho:N_\C\rightarrow \mathfrak{X}(X_\Delta)$ is defined in  \eqref{rho-def-eqn}.
Then $\pi$ is a holomorphic toric Poisson structure 
on $X_\Delta$ if and only if $\pi$ can be written as 
$$\pi=\sum_{1\leq i<j\leq n}a_{ij} v_i\wedge v_j ,$$ 
where $a_{ij}$ $(1\leq i<j\leq n)$ are complex constants.
\end{proposition}
\begin{proof}
\begin{itemize}
\item $\Leftarrow$:  Suppose that $$\pi=\sum_{1\leq i<j\leq n}a_{ij} v_i\wedge v_j,$$ 
where $a_{ij}$ $(1\leq i<j\leq n)$ are complex constants.
As $T_N$ is abelian, we have that $[v_i,v_j]=0$ for all $1\leq i<j\leq n$, 
which imply $[\pi,\pi]=0$. 
Obviously, $\pi$ is holomorphic and $T_N$-invariant. 
Thus $\pi=\sum_{1\leq i<j\leq n}a_{ij} v_i\wedge v_j$ is a holomorphic toric Poisson structure on $X_\Delta$.

\item $\Rightarrow$:  Suppose that $\pi$ is a holomorphic toric Poisson structure on $X_\Delta$. Then the restriction of $\pi$ on $T_N\subset X_\Delta$ is a holomorphic toric Poisson structure on $T_N$. We denoted it by $\tilde{\pi}$.
Any $T_N$-invariant holomorphic bi-vector field on $T_N\subset X_\Delta$ can be written as 
$$\sum_{1\leq i<j\leq n}a_{ij}\tilde{v_i}\wedge \tilde{v_j},$$
where $a_{ij}$ $(1\leq i<j\leq n)$ are complex constants, and $\tilde{v_i}$ $(1\leq i\leq n)$ are the restriction of
the vector fields $v_i$ $(1\leq i\leq n)$ on $T_N$.
Hence $\tilde{\pi}$ can be written as 
$$\tilde{\pi}=\sum_{1\leq i<j\leq n}a_{ij}\tilde{v_i} \wedge \tilde{v_j}.$$
As $T_N$ is a dense open set of $X_\Delta$, we have 
$$\pi=\sum_{1\leq i<j\leq n}a_{ij} v_i\wedge v_j.$$
\end{itemize}
\end{proof}

\begin{example}\label{Cn-exa2}
Let $X=\C^n$ and $(z_1,\ldots,z_n)$ be the standard coordinates on it. 
As we have shown in Example \ref{Cn-exa1}, $X=\C^n$ is a toric variety. 
Let $v_i=\rho(e_i)$ $(1\leq i\leq n)$.
Then we have $$v_i=z_i\frac{\partial}{\partial z_i}.$$
Any holomorphic toric Poisson structures on $X=\C^n$ can be written as
$$\pi=\sum_{1\leq i<j\leq n}a_{ij}v_i \wedge v_j,$$
where $a_{ij}$ $(1\leq i<j\leq n$ are complex constants. 
\end{example}

\begin{example}\label{CPn-exa2}
Let $X=\CP^n$ and let $[z_0,z_1,\ldots,z_n]$ be homogenous coordinates on it. 
As we have shown in Example \ref{CPn-exa1}, $X=\CP^n$ is a toric variety. 
Let $\mathcal{P}=\C^{n+1}\backslash\{0\}=\{(z_0,z_1,\ldots,z_n)\mid z_0,z_1,\ldots,z_n ~\text{are not all zeros}\} $, 
and let $p: \mathcal{P}=\C^{n+1}\backslash\{0\} \rightarrow\CP^n$ be the canonical projection. 
Then $v_i=p_{*}(z_i\frac{\partial}{\partial z_i})$ ($0\leq i\leq n$) are holomorphic vector fields on $X$, and $\sum_{i=0}^{n}v_i=0$. 
Moreover, by Equation \eqref{action-chi-eqn} and Equation \eqref{CPn-chi-eqn}, we have
$$v_i=\rho(e_i)\quad \text{for}\quad i=0,1,...,n.$$
Thus any holomorphic toric Poisson structures on $X$ can be written as
$$\pi=\sum_{1\leq i<j\leq n}a_{ij}v_i \wedge v_j,$$
where $a_{ij}$ $(1\leq i<j\leq n$ are complex constants. 
\end{example}

Given a holomorphic toric Poisson manifold $(X,\pi)$, as the torus action on $X$ is holomorphic and $\pi$ is invariant under the torus action, we have the following lemma.
\begin{lemma}\label{toric-PoCoh-lem}
For a holomorphic toric Poisson manifold $(X,\pi)$, the torus action on $X$ induces a torus action on the Poisson cohomology groups $H^\bullet_{\pi}(X)$.
\end{lemma}

\subsection{The standard Poisson structure on $\CP^n$}
In \cite{B-G-Y 06, G-Y 09}, Brown, Goodear and Yakimov studied
the geometry of the standard Poisson structures on affine spaces and flag varieties.
Let us review the definition of the standard Poisson structure on flag varieties.

Let $G$ be a connected complex reductive algebraic group with maximal torus $H$. 
Denote the corresponding Lie algebra by $\g$ and $\mathfrak{h}$. Denote 
$\Delta_{+}$ $(\Delta_{-})$
the set of all positive (negative) roots of $\g$ with respect to $\mathfrak{h}$.

The standard $r$-matrix of $\g$ is given by
\begin{equation}
r_\g=\sum_{\alpha\in\Delta_{+}} e_\alpha\wedge e_{-\alpha},
\end{equation}
where $e_\alpha$ and $e_{-\alpha}$ are root vectors of $\alpha$ and $-\alpha$, normalized by $\langle e_\alpha,f_\alpha\rangle=1$.
The standard Poisson structure on $G$ is given by 
$$\pi_G=L(r_\g)-R(r_\g),$$
where $L(r_\g)$ and $R(r_\g)$ refer to the left and right invariant bi-vector fields on $G$ associated to $r_\g\in\wedge^2\g\cong\wedge^2 T_e G$.

For a parabolic group $P$ containing $H$, $X=G/P$ is a flag variety.
The action of $G$ on $X=G/P$ induces a map $\mu:\g\rightarrow\mathfrak{X}(X)$. By abuse of notations, the induced maps 
$\wedge^k\g\rightarrow\mathfrak{X}^k(X)$ are also denoted by $\mu$.
The natural projection $$\phi: G\rightarrow X=G/P$$ 
induces the following Poisson structure on the flag variety $X=G/P$:
\begin{equation}
\pi_{st}=\phi_*(\pi_G)=\mu(r_\g),
\end{equation}
called \emph{the standard Poisson structure} on the flag varieties. The standard Poisson structure $\pi_{st}$ is a holomorphic Poisson structure on the flag variety $G/P$.

Next we will focus on the standard Poisson structure on $\CP^n$.

Let $G=GL(n+1,\C)$. Let $H$ be consisting of the diagonal matrices in 
$GL(n+1,\C)$ and $P$ consisting of matrices of the following form
$\begin{pmatrix}
\lambda &  b   \\
0 & D
\end{pmatrix}$, where $\lambda\in\C^*, b\in\C^n, D\in GL(n,\C)$.
Then $X=G/P$ becomes the projective space $\CP^n$.

The left action of $GL(n+1,\C)$ on $X=\CP^n$ can be written as:
$$A\cdot [z_0,z_1,...,z_n]\mapsto  p((z_0,z_1,...,z_n)A^t ), $$
where $A\in GL(n+1,\C)$, $ [z_0,z_1,...,z_n]\in\CP^n$, $ (z_0,z_1,...,z_n)\in\C^{n+1}$,
and $p$ is the canonical projection
$$ \C^{n+1}\backslash\{0\}\xrightarrow{p}\CP^n: 
(z_0,z_1,...,z_n)\rightarrow [z_0,z_1,...,z_n].$$

The standard $r$-matrix of $\g=gl(n+1,\C)$ can be written as
\begin{equation}
r_\g=\sum_{0\leq i<j\leq n}e_{ij}\wedge e_{ji},
\end{equation}
where $e_{ij}$ denotes the matrix having $1$ in the $(i+1,j+1)$ position and $0$ elsewhere.

Now we are ready to compute the standard Poisson structure on $\CP^n$.
\begin{proposition}\label{STPoisson-prop}
Let $X=\CP^n=GL(n+1,\C)/P$. Let $v_i=p_*(z_i\frac{\partial}{\partial z_i})$ $(i=0,1,...,n)$.
Then the standard Poisson structure on $X=\CP^n$ can be written as 
\begin{equation}
\pi_{st}=\sum_{1\leq i<j\leq n}v_i\wedge v_j.
\end{equation}
\end{proposition}
\begin{proof}
By computation, we have 
$$\mu(e_{ij})=p_*(z_{j}\frac{\partial}{\partial z_{i}}),\quad\text{and}\quad \mu(e_{ji})=p_*(z_{i}\frac{\partial}{\partial z_{j}})$$
for all $0\leq i<j\leq n$.

Hence the standard Poisson structure on $X=\CP^n$ can be written as 
\begin{align*}
\pi_{st}=\mu(r_\g)&=\sum_{0\leq i<j\leq n}\mu(e_{ij})\wedge \mu(e_{ji})\\
&=\sum_{0\leq i<j\leq n}p_*(z_j\frac{\partial}{\partial z_i})\wedge p_*(z_i\frac{\partial}{\partial z_j})\\
&=\sum_{0\leq i<j\leq n}p_*(z_i\frac{\partial}{\partial z_i})\wedge p_*(z_j\frac{\partial}{\partial z_j})\\
&=\sum_{0\leq i<j\leq n}v_i\wedge v_j\\
&=\sum_{1\leq i<j\leq n}v_i\wedge v_j.
\end{align*}
The last step holds as $\sum_{i=0}^n v_i=0$.
\end{proof}

\subsection{Exact sequences related to $\CP^n$}\label{Section-Euler}

\begin{theorem}\cite{Atiyah}
Let $\mathcal{P}$ be a principle bundle over $X$ with group $G$. Then there exists an exact sequence of vector bundles over $X$:
\begin{equation}\label{AtiyahSeq}
0\rightarrow \mathcal{P}\times_{G}\mathfrak{g}\rightarrow T\mathcal{P}/G\rightarrow TX\rightarrow 0,
\end{equation}
where $\mathcal{P}\times_{G}\mathfrak{g}$ is the bundle associated to $\mathcal{P}$ by the adjoint representation of $G$ on $\mathfrak{g}=Lie(G)$, 
and $T\mathcal{P}/G$ is the bundle of invariant vector fields on $\mathcal{P}$.
\end{theorem}
Recall that for a principle $G$-bundle $\mathcal{P}$ over $X$, and a representation of $G$ on a vector space $V$,  the associated vector bundle over $X$ is defined to be $\mathcal{P}\times_{G}V=(\mathcal{P}\times V)/\sim$, where $(x.g, v)\sim((x, g.v)$,  $\forall x\in \mathcal{P}$, $g\in G, v\in V$. 

Let $\mathcal{\mathcal{P}}=\C^{n+1}\backslash\{0\}$, $X=\CP^n$, and let 
$p: \mathcal{P}=\C^{n+1}\backslash\{0\}\rightarrow\CP^n$ be the canonical projection. 
The group $G=\C^*$ operates by right multiplication on $\mathcal{P}=\C^{n+1}\backslash\{0\}$:
$$ \lambda: v\rightarrow v\lambda, \quad v\in\C^{n+1}\backslash\{0\}, \quad \lambda\in\C^*.$$
Then $\mathcal{P}$ is a principle $\C^*$-bundle over $X$.

As $G=\C^*$ is abelian, the adjoint representation is trivial, we have $\mathcal{P}\times_{G}\mathfrak{g}\cong X\times\C$. 

The $\C^*$-action on $T\mathcal{P}\cong \C^{n+1}\backslash\{0\}\times\C^{n+1}$ is given by:
$$ (x\times v)\lambda=x\lambda\times v\lambda, \qquad x\in \C^{n+1}\backslash\{0\}, v\in\C^{n+1}, \lambda\in\C^*.$$
Hence we have 
$T\mathcal{P}/G\cong \mathcal{P}\times_{\C^*}\C^{n+1}\cong O(1)^{\oplus(n+1)}$, 
where $\mathcal{P}\times_{\C^*}\C^{n+1}$ is the associated bundle of $\mathcal{P}$ 
by the $\C^*$ representation on $\C^{n+1}$ given by:
$$\rho(\lambda)v=\lambda^{-1}v,\qquad v\in\C^{n+1},\lambda\in\C^*.$$ 
In this case, the Atiyah exact sequence \eqref{AtiyahSeq} becomes
\begin{equation}\label{EulerSeq}
0\rightarrow\C\rightarrow O(1)^{\oplus(n+1)}\rightarrow T\CP^{n}\rightarrow 0,
\end{equation}
which is exactly the Euler exact sequence of $\CP^n$.

Suppose that $(z_0, z_1,\ldots, z_n)$ are the standard coordinates on $\C^{n+1}$. 
Then the canonical projection $p: \C^{n+1}\backslash\{0\}\rightarrow\CP^n$ becomes
$(z_0, z_1,\ldots, z_n)\rightarrow [z_0, z_1,\ldots, z_n]$.
The map $p: \mathcal{P}\rightarrow\CP^n$ induces a map $T\mathcal{P}/G\xrightarrow{p_*} T\CP^n$. 
By identification of $T\mathcal{P}/G$ with $O(1)^{\oplus(n+1)}$, and by abuse of notation, we also denote $p_*$ by the map $O(1)^{\oplus(n+1)}\rightarrow T\CP^{n}$ in the Euler sequence \eqref{EulerSeq}. 
Let us denote $E$ by the vector bundle $O(1)^{\oplus(n+1)}$, denote $L$ by 
the kernel of the map $O(1)^{\oplus(n+1)}\xrightarrow{p_*} T\CP^{n}$.
Then $L=\C\overrightarrow{e}$ is a subbundle of $E$, where 
$\overrightarrow{e}=\sum_{i=0}^{n}z_i \frac{\partial}{\partial z_i}$ is  the Euler vector field. 
The Euler exact sequence \eqref{EulerSeq} can then be written as 
\begin{equation}\label{EulerSeq2}
0\rightarrow L \hookrightarrow E\xrightarrow{p_*} TX\rightarrow 0.
\end{equation}
Let $L\wedge(\wedge^j E)=\C\overrightarrow{e}\wedge(\wedge^{j}E)$ $(0\leq j\leq n-1)$, that is a subbundle of $\wedge^{j+1}E$. Then we have the following lemma, 
which is known to Bondal \cite{Bondal}.
\begin{lemma}\label{EulerSeq-lem}
For $X=\CP^n$. we have the exact sequences
\begin{equation}\label{EulerSeq-lem-eqn1}
0\rightarrow L\wedge(\wedge^{j-1} E)\hookrightarrow \wedge^{j}E\xrightarrow{\overrightarrow{e}\wedge\cdot} L\wedge(\wedge^j E)\rightarrow 0
\end{equation}
and 
\begin{equation}\label{EulerSeq-lem-eqn2}
0\rightarrow L\wedge(\wedge^{j-1} E)\hookrightarrow\wedge^{j}E\xrightarrow{p_*} \wedge^j TX\rightarrow 0
\end{equation}
for all $j\geq 1$,  where
\begin{itemize} 
\item [(I)] $L\wedge(\wedge^{j-1} E)\hookrightarrow\wedge^{j}E$ is the embedding of 
$L\wedge(\wedge^{j-1} E)$ as a subbundle of $\wedge^{j}E$; 
\item [(II)] $\wedge^{j}E\xrightarrow{\overrightarrow{e}\wedge\cdot} L\wedge(\wedge^j E)$ is defined by the wedge of $\overrightarrow{e}$ with elements in $\wedge^{j}E$;
\item [(II)] $\wedge^{j}E\xrightarrow{p_*} \wedge^j TX$ is induced by the map 
$E\xrightarrow{p_*} TX$ in Equation \eqref{EulerSeq2}.
\end{itemize}
\end{lemma}
\begin{proof}

At any point $x\in X$, for any $\alpha_x\in\wedge^j E\mid_x$, 
we have that $\overrightarrow{e}_x\wedge \alpha_x=0$
if and only if there exist $\beta_x\in\wedge^{j-1} E\mid_x$, such that
$\alpha_x=\overrightarrow{e}_x\wedge\beta_x.$
Thus \eqref{EulerSeq-lem-eqn1} is an exact sequence for all $j\geq 1$.

As we have shown in \eqref{EulerSeq2},  the kernel of $E\xrightarrow{p_*} TX$ is the trivial bundle $L=\C\overrightarrow{e}$. 
At any point $x\in X$, the kernel of
$E\mid_x\xrightarrow{p_*} TX\mid_x$ is $\C\overrightarrow{e}_x$. 
As a consequence, the kernel of
$\wedge^{j}E\mid_x \xrightarrow{p_*} \wedge^j TX\mid_x$
 is
$$\overrightarrow{e}_x\wedge (\wedge^{j-1}E\mid_x).$$
Thus \eqref{EulerSeq-lem-eqn2} is an exact sequence for all $j\geq 1$.
\end{proof}

\section{The cohomology groups $H^{i}(\CP^n,\wedge^{j} \mathcal{T}_{\CP^n})$}\label{section3}

\subsection{The vanishing of the cohomology group 
$H^{i}(\CP^n,\wedge^{j} \mathcal{T}_{\CP^n})$ for $i>0$ and $0\le j\le n$}

\begin{theorem}\label{TheoemVanish}
For $X=\CP^n$,  we have 
\begin{equation}
H^{i}(X,\wedge^{j} \mathcal{T}_X)=0
\end{equation}
for all $i>0$ and $0\le j\le n$.
\end{theorem}

\begin{remark}
Theorem \ref{TheoemVanish} should be known to Bott \cite{Bott}. 
In the case of $j=1$, the conclusion $H^{i}(X, \mathcal{T}_X)=0$ $(i>0)$ is a special case of the Theorem VII in \cite{Bott}. 
\end{remark}

To make the paper self contained, we will give a proof of Theorem \ref{TheoemVanish}.

Let $L=\C\overrightarrow{e}$ and $E=O(1)^{\oplus(n+1)}$ as in Section \ref{Section-Euler}. By Lemma \ref{EulerSeq-lem}, we get the following lemma.

\begin{lemma}\label{Euler-Seq-lem2}
For $X=\CP^n$, we have
\begin{equation}
H^i(X, \wedge^{j}TX)\cong H^i(X, L\wedge(\wedge^{j}E))\cong H^{i+1}(X,L\wedge(\wedge^{j-1}E))
\end{equation}
for all $i>0$ and $j\geq 1$.
\end{lemma}

\begin{proof}
The exact sequence \eqref{EulerSeq-lem-eqn1} in Lemma \ref{EulerSeq-lem}
 induces a long exact sequence
\begin{equation}\label{EulerSeq-lem-eqn3}
\ldots\rightarrow H^i(X, \wedge^j E)\rightarrow H^i(X, L\wedge(\wedge^{j}E))\rightarrow
H^{i+1}(X,L\wedge(\wedge^{j-1}E))\rightarrow H^{i+1}(X, \wedge^j E)\rightarrow\ldots.
\end{equation}
As $\wedge^j E=\wedge^j(O(1)^{\oplus(n+1)})\cong O(j)^{\oplus{n+1\choose j }},$
we have that
\begin{align*}
H^i(X,\wedge^j E)&=H^i(X,O(j)^{\oplus{n+1\choose j }})\\
&=H^i(X,O(j))^{\oplus{n+1\choose j }}\\
&=H^i(X,K_X\otimes O(n+1+j))^{\oplus{n+1\choose j }},
\end{align*}
where $K_X\cong O(-n-1)$ is the canonical line bundle of $X=\CP^n$.
By Kodaira vanishing theorem, we have $H^i(X,K_X\otimes O(n+1+j))=0$ for $i>0$, 
which implies $H^i(X,\wedge^j E)=0$ for $i>0$.
Thus the exact sequence \eqref{EulerSeq-lem-eqn3} induces the isomorphism
\begin{equation*}\label{EulerSeq-lem-eqn4}
H^i(X, L\wedge(\wedge^{j}E))\cong H^{i+1}(X,L\wedge(\wedge^{j-1}E))
\end{equation*}
for all $i>0$ and $j\geq 1$.

By the similar reason, the exact sequence \eqref{EulerSeq-lem-eqn2} in Lemma \ref{EulerSeq-lem} induces the isomorphism
\begin{equation*}\label{EulerSeq-lem-eqn5}
H^i(X, \wedge^{j}TX)\cong H^{i+1}(X,L\wedge(\wedge^{j-1}E))
\end{equation*}
for all $i>0$ and $j\geq 1$.

Thus we proved the lemma.
\end{proof}

{\bf Proof of Theorem \ref{TheoemVanish}:}
\begin{proof}
\begin{enumerate}
\item In the case of $j=0$, $H^i(X,\mathcal{O}_X)=0~(i>0)$ is a well known result. 
\item In the case of $j\geq 1$ and $i>0$, by Lemma \ref{Euler-Seq-lem2}, we have that
$$H^i(X, \wedge^{j}TX)\cong H^i(X, L\wedge(\wedge^{j}E))\cong H^{i+1}(X,L\wedge(\wedge^{j-1}E))\cong\cdots \cong H^{i+j}(X, L).$$

As $L$ is the trivial line bundle, we have $H^{i+j}(X, L)=H^{i+j}(X, O_X)=0$ for 
$i>0$ and $j\geq 1$, which implies that
$$H^{i}(X,\wedge^{j} \mathcal{T}_X)=H^{i+j}(X, L)=0$$
for all $i>0$ and $j\geq 1$. 
\end{enumerate}
\end{proof}

\subsection{Holomorphic vector fields and multi-vector fields on $\CP^n$}\label{section-multivector}

Let us introduce some notations first, which are important for this paper. Our setting are based on Example \ref{CPn-exa1} and Example \ref{CPn-exa2}.
\begin{itemize}
\item Let $v_i=p_{*}(z_i\frac{\partial}{\partial z_i})=\rho(e_i)$ $(0\leq i\leq n)$. 
Denote $W$ as the $n$-dimensional $\C$-vector space 
generated by $\{v_1,...,v_n\}$. Then we have $W=\rho(N_\C)$.
Let $W^k=\wedge^k W$ $(1\leq k\leq n)$ and $W^0=\C$. 
Then $W^k$ can be considered as a subspace of $H^0(X,\wedge^k \mathcal{T}_X)$.

\item Given any $I=(m_1,...,m_n)\in M$, let $m_0=-\sum_{i=1}^{n}m_i$, 
then $\chi^m=z_0^{m_0}...z_n^{m_n}$  can be considered as a rational function on $\CP^n$.
Suppose that $\{m_{i_1},m_{i_2},...,m_{i_l}\mid 0\leq i_1<i_2<...<i_l\leq n\}$ is the set of all the elements equal to $-1$ in $\{m_0, m_1,\ldots, m_n\}$. 
Set
\begin{align*}
|I|=l, \quad \chi^I=z_0^{m_0}...z_n^{m_n}, \quad
\cV_I=v_{i_1}\wedge...\wedge v_{i_l}\in W^l, \quad
\cE_I=e_{i_1}\wedge...\wedge e_{i_l}\in\wedge^l N.
\end{align*}
Notice that $v_0=-\sum_{i=1}^{n}v_i$ and $e_0=-\sum_{i=1}^{n}e_i$.
 \end{itemize}
 
The following theorem described the structure of the space of holomorphic vector fields and multi-vector fields on $\CP^n$ with respect to the $(\C^*)^n$-action on $H^0(X,\wedge^k \mathcal{T}_X)$ $(0\leq k\leq n)$.

\begin{theorem}\label{CPn-multiVec-thm}
Let $X=\CP^n$. 
Let  $S_k$ $(0\leq k\leq n)$ be the set consisting of all $I\in M$ satisfying the conditions
\begin{equation} \label{I-con1}
\langle I,e_i\rangle=m_i\geq -1 \quad (0\leq i\leq n)
\end{equation}
and 
\begin{equation}\label{I-con2}
|I|\leq k.
\end{equation}
Let 
\begin{equation}
V_{I}^k=\C (\chi^I\cdot \cV_I)\wedge W^{k-|I|}
\end{equation}
for $|I|\leq k\leq n$.
The $(\C^*)^n$-action on $X=\CP^n$ induced a $(\C^*)^n$-representation on $H^0(X,\wedge^k \mathcal{T}_X)$,
which has the weight spaces decomposition 
\begin{equation}
H^0(X,\wedge^k \mathcal{T}_X)=\bigoplus_{I\in S_k} V_{I}^k,
\end{equation}
where $S_k$ is the set of all weights,
and $V_I^k$ is the weight space of the weight $I$.
\end{theorem}

In Theorem \ref{CPn-multiVec-thm}, $S_k$ is the set of all $I\in M$ satisfying conditions \ref{I-con1} and \ref{I-con2}. Let us denote $S(i)$ $(0\leq i\leq n)$ as the set of all all $I\in M$ satisfying the condition \eqref{I-con1} and $|I|=i$. 
Then we have $S_k=S_{k-1}\cup S(i)$ and
$S_0\subseteq S_1\subseteq S_2\ldots\subseteq S_n.$
By Theorem \ref{CPn-multiVec-thm}, we have

\begin{proposition}\label{CPn-multiVec-prop}
Let $X=\CP^n$. We have
\begin{equation}
H^0(X,\wedge^{k} \mathcal{T}_X)=(H^0(X,\wedge^{k-1} \mathcal{T}_X)\wedge W)\oplus
(\bigoplus_{I\in S(k)}\C \chi^I\cdot \cV_I) \quad (1\leq k\leq n),
\end{equation}
where $S(k)$ is the set of all all $I\in M$ satisfying the condition \eqref{I-con1} and $|I|=k$. 
\end{proposition}

Next we will prove Theorem \ref{CPn-multiVec-thm} and Proposition 
\ref{CPn-multiVec-prop}. 

By Lemma \ref{EulerSeq-lem}, we have

\begin{lemma}\label{MultiVectLem}
For $X=\CP^n$, we have the exact sequences
\begin{equation}\label{MultiVectLem-eqn}
0\rightarrow H^0(X, L\wedge(\wedge^{j-1} E))\rightarrow H^0(X,\wedge^{j}E)\xrightarrow{p_*} H^0(X,\wedge^j TX)\rightarrow 0
\end{equation}
for all $j\geq 1$.
\end{lemma}

\begin{proof}
The exact sequence \eqref{EulerSeq-lem-eqn2} induces a long exact sequence
$$0\rightarrow H^0(X, L\wedge(\wedge^{j-1} E))\rightarrow H^0(X,\wedge^{j}E)\xrightarrow{p_*} H^0(X,\wedge^j TX)\rightarrow H^1(X, L\wedge(\wedge^{j-1} E))\rightarrow\cdots$$
for all $j\geq 1$.
By Lemma \ref{Euler-Seq-lem2}, we have
$$H^1(X, L\wedge(\wedge^{j-1} E))\cong H^2(X, L\wedge(\wedge^{j-2} E))\cong\cdots\cong H^j(X, L)=H^j(X, O_X)=0$$
for all $j\geq 1$.
Therefore we have the exact sequences
$$0\rightarrow H^0(X, L\wedge(\wedge^{j-1} E))\rightarrow H^0(X,\wedge^{j}E)\xrightarrow{p_*} H^0(X,\wedge^j TX)\rightarrow 0$$
for all $j\geq 1$.
\end{proof}
\begin{remark}
In the case of $j=1$, the exact sequence \eqref{MultiVectLem-eqn} becomes
\begin{equation}\label{MultiVectEq1}
 0\rightarrow\C\rightarrow H^{0}(X,O(1))^{\oplus(n+1)}\rightarrow H^{0}(X, T_X)\rightarrow 0.
 \end{equation}
\end{remark}

As $\wedge^j E=O(j)^{\oplus{n+1\choose j }}$, we have that
 $H^0(X,\wedge^{j}E)\cong H^0(X, O(j))^{\oplus{n+1\choose j }}$. 

Let us denote $F_j$ $(1\le k\le n+1)$ by the complex vector space of the $j$-vector fields
$$\sum_{0\le i_1<i_2<\ldots< i_k\le n}g_{i_1,i_2,\ldots,i_j}\frac{\partial}{\partial z_{i_1}}\wedge\ldots\wedge\frac{\partial}{\partial z_{i_j}}$$
on $\C^{n+1}$, where $g_{i_1,i_2,\ldots,i_j}$ are homogenous polynomials with variables 
$z_0,z_1,...,z_n$ of degree $j$.
The restriction of such $j$-vector fields on $\mathcal{P}=\C^{n+1}\backslash\{0\}\subset\C^{n+1}$ form a vector space. We denote it by $\tilde{F_j}$.
Then the complex vector spaces below are isomorphic:
\begin{itemize}
\item[(a)] $F_j$;
\item[(b)] $\tilde{F_j}$;
\item[(c)] $H^0(X,\wedge^{j}E)$;
\item[(d)] $H^{0}(X, O(j))^{\oplus{n+1\choose j }}$;
\end{itemize}
where $1\le j\le n+1$.

By identification of the vector spaces above, and by Lemma \ref{MultiVectLem}, we have

\begin{lemma} \cite{Bondal} \label{MultiVectLem3}
Let $p: \mathcal{P}=\C^{n+1}\backslash\{0\} \rightarrow\CP^n: (z_0,z_1,\ldots,z_n)\rightarrow [z_0,z_1,\ldots,z_n]$ be the canonical projection.
Then we have
\begin{enumerate} 
\item Any holomorphic $k$-vector field on $\CP^n$ can be written as 
$$p_*(\sum_{0\le i_1<i_2<\ldots< i_k\le n}g_{i_1,i_2,\ldots,i_k}\frac{\partial}{\partial z_{i_1}}\wedge\ldots\wedge\frac{\partial}{\partial z_{i_k}}),$$
where $g_{i_1,i_2,\ldots,i_k}$ are homogenous polynomials with variables 
$z_0,z_1,...,z_n$ of degree $k$.
Or in other words, $$H^0(X,\wedge^k \mathcal{T}_X)=p_*(\tilde{F_k}).$$
\item 
The kernel of the map $p_*: \tilde{F_k}\rightarrow H^0(X,\wedge^k \mathcal{T}_X)$ is
$$\ker p_{*}=(\sum_{i=0}^{n}z_{i}\frac{\partial}{\partial z_{i}})\wedge\tilde{F}_{k-1}.$$
\end{enumerate}
\end{lemma}

{\bf Proof of Theorem \ref{CPn-multiVec-thm}:}
\begin{proof}
\begin{enumerate}
\item First, we will prove that
\begin{equation}\label{Multi-sum-eqn}
H^0(X,\wedge^k \mathcal{T}_X)=\sum_{I\in S_k}V_I^k\quad (0\leq k\leq n),
\end{equation}
which is equivalent to 
$$ H^0(X,\wedge^k \mathcal{T}_X)\subseteq\sum_{I\in S_k}V_I^k$$
and 
$$\sum_{I\in S_k}V_I^k\subseteq H^0(X,\wedge^k \mathcal{T}_X).$$

\begin{enumerate}
\item
By Lemma \ref{MultiVectLem3}, 
any holomorphic $k$-vector field $\Psi$ on $\CP^n$ can be written as 
\begin{align*}
\Psi&=p_*(\sum_{0\le i_1<i_2<\ldots< i_k\le n}g_{i_1,i_2,\ldots,i_k}\frac{\partial}{\partial z_{i_1}}\wedge\ldots\wedge\frac{\partial}{\partial z_{i_k}})\\
&=\sum_{0\le i_1<i_2<\ldots< i_k\le n}f_{i_1,i_2,\ldots,i_k}v_{i_1}\wedge\ldots\wedge v_{i_k},
\end{align*}
where $g_{i_1,i_2,\ldots,i_k}$ are homogenous polynomials with variables 
$z_0,z_1,...,z_n$ of degree $k$, and $f_{i_1,i_2,\ldots,i_k}=\displaystyle{\frac{g_{i_1,i_2,\ldots,i_k}}{\prod_{s=1}^{k}z_{i_s}}}$.
Suppose that 
\begin{equation}\label{f^m_i}
f_{i_1,i_2,\ldots,i_k}=\sum c_{m_0,...,m_n}^{i_1,...,i_k} z_0^{m_0}...z_n^{m_n},
\end{equation}
where $c_{m_0,...,m_n}^{i_1,...,i_k}$ are complex constants.

The integers $m_i$ $(0\leq i\leq n)$ in Equation \eqref{f^m_i} have to satisfy the conditions below:
\begin{enumerate}
\item $m_{i_s}\geq -1$\quad for \quad $1\leq s\leq k$,
\item $m_j\geq 0$\quad for \quad $j\notin \{i_s\mid 1\leq s\leq k\},$
\item $\sum_{i=0}^{n}m_i=0.$
\end{enumerate}

It is easy to verify that any element $I=(m_1,...,m_n)\in M$ satisfying 
the conditions (i), (ii) and (iii) satisfies $I\in S_k$. 

Without loss of generality, 
suppose that $\{m_{i_1},m_{i_2},...,m_{i_l}\mid 0\leq i_1<i_2<...<i_l\leq n\}$ is the set of 
all the elements equal to $-1$ in $\{m_0, m_1,\ldots, m_n\}$. 
Then
\begin{align*}
z_0^{m_0}...z_n^{m_n}v_{i_1}\wedge\ldots\wedge v_{i_k}&=\chi^I \cdot(v_{i_1}\wedge\ldots\wedge v_{i_l})\wedge (v_{i_{l+1}}\wedge...\wedge v_{i_k})\\
&=\chi^I\cdot\cV_I\wedge (v_{i_{l+1}}\wedge...\wedge v_{i_k})
\end{align*}
is a $k$-vector field in the space 
$V_I^k=V_I^k$ ($I\in S_k$).
As
\begin{equation} \label{Xi-eqn}
\Psi=\sum c_{m_0,...,m_n}^{i_1,...,i_k} z_0^{m_0}...z_n^{m_n}
v_{i_1}\wedge\ldots\wedge v_{i_k},
\end{equation}
we have that
$$\Psi\in \sum_{I\in S_k}V_I^k\quad (0\leq k\leq n).$$

Thus
\begin{equation*}
H^0(X,\wedge^k \mathcal{T}_X)\subseteq\sum_{I\in S_k}V_I^k\quad (0\leq k\leq n).
\end{equation*}

\item
On the other hand, given an element $I=(m_1,...,m_n)\in S_k\subset M$, 
suppose that $\{m_{i_1},m_{i_2},...,m_{i_l}\mid 0\leq i_1<i_2<...<i_l\leq n\}$ is the set of all the elements equal to $-1$ in $\{m_0, m_1,\ldots, m_n\}$.  Then we have
\begin{align*}
\chi^I\cdot \cV_I&=z_0^{m_0}...z_n^{m_n}v_{i_1}\wedge\ldots\wedge v_{i_l}\\
&=z_0^{m_0}...z_n^{m_n}p_*(z_{i_1}\frac{\partial}{\partial z_{i_1}})\wedge\ldots\wedge 
p_*(z_{i_l}\frac{\partial}{\partial z_{i_l}})\\
&=p_*((z_{i_1}...z_{i_l}).(z_0^{m_0}...z_n^{m_n})\frac{\partial}{\partial z_{i_1}}\wedge\ldots\wedge\frac{\partial}{\partial z_{i_l}})\\
&=p_*((\prod_{i\notin\{i_1,...i_l\}}z_i^{m_i})\frac{\partial}{\partial z_{i_1}}\wedge\ldots\wedge\frac{\partial}{\partial z_{i_l}}).
\end{align*}
As $m_i\geq 0$ for $i\notin\{i_1,...i_l\}$, we know that $\prod_{i\notin\{i_1,...i_l\}}z_i^{m_i}$ is a polynominal with variables $z_0, z_1,...z_n$ of degree $l$.  By Lemma \ref{MultiVectLem3}, $\chi^I\cdot \cV_I$ is a holomorphic $l$-vector field on $X=\CP^n$. 

As $W^{k-|I|}=W^{k-l}$ is a subspace of $H^0(X,\wedge^{k-l} \mathcal{T}_X)$, 
we get that
$V_I^k=\C (\chi^I\cdot \cV_I)\wedge W^{k-|I|}$ is a subspace of $H^0(X,\wedge^{k} \mathcal{T}_X)$.

Therefore we have
\begin{equation*}
\sum_{I\in S_k}V_I^k\subseteq H^0(X,\wedge^k \mathcal{T}_X).
\end{equation*}

\item By the arguments above, we proved Equation \eqref{Multi-sum-eqn}.
\end{enumerate}

\item By Equation \eqref{Multi-sum-eqn}, it is easy to verify that 
$V_I^k=\C (\chi^I\cdot \cV_I)\wedge W^{k-|I|}$ is the weight space of the weight $I$.
Thus we have
\begin{equation*}
H^0(X,\wedge^k \mathcal{T}_X)=\bigoplus_{I\in S_k}V_I^k\quad (0\leq k\leq n).
\end{equation*}
\end{enumerate}
\end{proof}

{\bf Proof of Proposition \ref{CPn-multiVec-prop}: }
\begin{proof}
Since $S_k=S_{k-1}\cup S(k)$, by Theorem \ref{CPn-multiVec-thm}, we have
\begin{gather*}
H^0(X,\wedge^k \mathcal{T}_X)=\bigoplus_{I\in S_k}\C (\chi^I\cdot \cV_I)\wedge W^{k-|I|}\\
=(\bigoplus_{I\in S_{k-1}}\C (\chi^I\cdot \cV_I)\wedge W^{k-|I|}) \oplus
(\bigoplus_{I\in S(k)}\C \chi^I\cdot W^k).
\end{gather*}
As $W^{k-|I|}=W^{k-1-|I|}\wedge W$, by Theorem \ref{CPn-multiVec-thm}, we have 
\begin{gather*}
\bigoplus_{I\in S_{k-1}}\C (\chi^I\cdot \cV_I)\wedge W^{k-|I|}
=(\bigoplus_{I\in S_{k-1}}\C (\chi^I\cdot \cV_I)\wedge W^{k-1-|I|})\wedge W\\
=H^0(X,\wedge^{k-1} \mathcal{T}_X)\wedge W.
\end{gather*}
Therefore we have
$$H^0(X,\wedge^{k} \mathcal{T}_X)=(H^0(X,\wedge^{k-1} \mathcal{T}_X)\wedge W)\oplus
(\bigoplus_{I\in S(k)}\C \chi^I\cdot W^k) \quad (1\leq k\leq n).$$
\end{proof}

\section{Poisson cohomology of $\CP^n$}\label{section4}

\subsection{Poisson cohomology of toric Poisson structures on $\CP^n$}
Let us introduce some necessary notations first. 
\begin{itemize}
\item  Let $\pi=\sum_{1\leq i<j\leq n}a_{ij} v_i\wedge v_j$ 
be a holomorphic toric Poisson structure on $X=\CP^n$, 
where $v_i=p_*(z_i\frac{\partial}{\partial z_i})=\rho(e_i)$ $(0\leq i\leq n)$.
We define $\Pi\in\wedge^2 N_\C$ by  the equation
\begin{equation}
\rho(\Pi)=\pi.
\end{equation}
Then 
$$\Pi=\sum_{1\leq i<j\leq n}a_{ij} e_i\wedge e_j.$$
\item Given $I=(m_1,...,m_n)\in M$, we have
$$\imath_{I}\Pi\in N_\C=N\otimes\C,$$ 
where $\imath_{I}\Pi$ denotes the contraction of $I\in M$ with $\Pi\in\wedge^2 N_\C$.
Let $m_0=-\sum_{i=1}^{n}m_i$. Then $\chi^I=z_0^{m_0} z_1^{m_1}\cdots z_n^{m_n}$ can be considered as a rational function on $X=\CP^n$.
\end{itemize}

By Lemma  \ref{LSX-lem} and Theorem \ref{TheoemVanish},  
the Poisson cohomology of $(X=\CP^n,\pi)$ is isomorphic to the cohomology of the complex
\begin{equation}
H^{0}(X,\mathcal{O}_{X})\xrightarrow {d_{\pi}}
H^{0}(X,T_{X})\xrightarrow{d_{\pi}}H^{0}(X,\wedge^{2}T_{X})
\xrightarrow{d_{\pi}}
\ldots \xrightarrow{d_{\pi}}H^{0}(X,\wedge^{n}T_{X})
\end{equation}
where $d_{\pi}=[\pi,\cdot]$. 
By Lemma \ref{toric-PoCoh-lem}, there is a $(\C^*)^n$-action on $H^{\bullet}_{\pi}(X)$. The following theorem described the Poisson cohomology group $H^{\bullet}_{\pi}(X)$ by the decomposition of the weight spaces of
the $(\C^*)^n$-respresentation on $H^{\bullet}_{\pi}(X)$. 

\begin{theorem}\label{CPn-cohomology-thm}
Let $\pi$ be a holomorphic toric Poisson structure on $X=\CP^n$. 
Let $V_I^k=\C (\chi^I\cdot \cV_I)\wedge W^{k-|I|}$ for $|I|\leq k\leq n$.
Then we have
\begin{enumerate}
\item for $0\leq k\leq n$, 
\begin{equation}\label{CPn-cohomology-thm-eqn1}
H_{\pi}^k(X)=\bigoplus_{I\in S_k(\pi)}V_I^k,
\end{equation}
where $S_k(\pi)$ is the set consisting of all $I\in M$ satisfying 
\begin{align}
\langle I,e_i\rangle=&m_i\geq -1\quad (0\leq i\leq n); \label{CPn-cohomology-thm-eqn2}\\  
&|I|\leq k; \label{CPn-cohomology-thm-eqn3}
\end{align}
and the equation
\begin{equation}\label{mI-eqn}
(\imath_{I}\Pi)\wedge \cE_I=0.
\end{equation}
\item  $H_{\pi}^k(X)=0$ for $k>n$.
\end{enumerate}
\end{theorem}

\begin{remark}\label{CPn-cohomology-rem}
\begin{enumerate}
\item $S_k(\pi)$ is a subset of $S_k$ consisting of all $I\in S_k$ satisfying Equation \eqref{mI-eqn}. By Theorem \ref{CPn-multiVec-thm},
$\bigoplus_{I\in S_k(\pi)}V_I^k$ 
is a subspace of $H^{0}(X,\wedge^{k}T_{X})$.  
In Theorem \ref{CPn-cohomology-thm}, the elements in 
$\bigoplus_{I\in S_k(\pi)}V_I^k$ 
represent the cohomology classes in the quotient space
$$\frac{\ker: H^{0}(X,\wedge^{k}T_{X})\xrightarrow{d_{\pi}}H^{0}(X,\wedge^{k+1}T_{X})}{\Ima: H^{0}(X,\wedge^{k-1}T_{X})\xrightarrow{d_{\pi}}H^{0}(X,\wedge^{k}T_{X})}.$$
\item For $k=0$, $H_\pi^0(X)=\C$ consist of complex constants.
\item For $k=n$, Theorem \ref{CPn-cohomology-thm} can be state in the following equivalent way:
\begin{equation}
H_{\pi}^n(X)=\oplus_{I}\C \chi^I \cdot v_1\wedge...\wedge v_n
\end{equation}
for all $I\in S_n\subset M$ satisfying one of the following conditions
\begin{align}
\begin{cases}
|I|=n,\\
(\imath_{I}\Pi)\wedge \cE_I=0.
\end{cases}
\end{align}
\end{enumerate}
\end{remark}

Let us denoted $S(i,\pi)$ as the set of all $I\in M$ satisfying $|I|=i$ and the conditions \eqref{CPn-cohomology-thm-eqn2} and \eqref{mI-eqn}.
By the similar method as we have done in Proposition \ref{CPn-multiVec-prop}, 
we can prove the following proposition.

\begin{proposition}\label{CPn-cohomology-prop}
Let $\pi$ be a holomorphic toric Poisson structure on $X=\CP^n$. 
For $1\leq k\leq n$, we have
\begin{equation}\label{CPn-cohomology-thm-eqn1}
H_{\pi}^k(X)=(H_{\pi}^{k-1}(X)\wedge W)\oplus(\bigoplus_{I\in S(k,\pi)}\C (\chi^I\cdot \cV_I)),
\end{equation}
where $S(k,\pi)$ is the set of all $I\in M$ satisfying $|I|=k$ and the conditions \eqref{CPn-cohomology-thm-eqn2} and \eqref{mI-eqn}.
\end{proposition}

We need some lemmas to prove Theorem \ref{CPn-cohomology-thm}.
\begin{lemma}\label{Poisson-Coh-lem1}
Let $\pi=\sum_{1\leq i<j\leq n}a_{ij} v_i\wedge v_j$ be a holomorphic toric Poisson structure on $X=\CP^n$. 
For any $I\in M$, we have that
\begin{gather}
[\pi,\chi^I]=\chi^I\cdot\rho(\imath_{I}\Pi) \label{PoCoh-lemma-eqn1};\\
[\pi, \chi^I \cdot \cV_I]=\rho(\imath_{I}\Pi)\wedge(\chi^I\cdot \cV_I)
=\chi^I\cdot\rho((\imath_{I}\Pi)\wedge \cE_I)\label{PoCoh-lemma-eqn2};\\
[\pi,\chi^I\cdot \cV_I\wedge w]=\rho(\imath_{I}\Pi)\wedge(\chi^I\cdot \cV_I\wedge w)\label{PoCoh-lemma-eqn2.5};
\end{gather}
where $w\in W^{k-|I|}$ and $|I|\leq k\leq n$.
\end{lemma}

\begin{proof}
\begin{itemize}
\item [(I)] For any $I\in M$,  we have
\begin{equation}\label{PoCoh-lemma-eqn3}
[v_i, \chi^I]=v_i(\chi^I)=\langle e_i,I\rangle \chi^I 
\end{equation}
for $0\leq i\leq n$.
By Equation \eqref{PoCoh-lemma-eqn3},  we can prove Equation \eqref{PoCoh-lemma-eqn1} with a simple computation.

\item [(II)]  As $T_N\cong(\C^*)^n$ is commutative, we have
\begin{equation}\label{PoCoh-lemma-eqn4}
[v_i,v_j]=0
\end{equation}
 for $0\leq i,j\leq n$.

By Equation \eqref{PoCoh-lemma-eqn1} and Equation \eqref{PoCoh-lemma-eqn4}, we have
\begin{equation*}
[\pi, \chi^I \cdot \cV_I]=\rho(\imath_{I}\Pi)\wedge(\chi^I\cdot \cV_I).
\end{equation*}

As $\rho(\cE_I)=\cV_I$, we have 
\begin{gather*}
[\pi, \chi^I \cdot \cV_I]=\rho(\imath_{I}\Pi)\wedge (\chi^I\cdot \cV_I)
=\chi^I\cdot\rho((\imath_{I}\Pi)\wedge \cE_I).
\end{gather*}

\item [(III)] As a consequence of Equation \eqref{PoCoh-lemma-eqn2} and Equation \eqref{PoCoh-lemma-eqn4}, we can get Equation \eqref{PoCoh-lemma-eqn2.5}. 
\end{itemize}
\end{proof}

By Theorem \ref{CPn-multiVec-thm} and Lemma \ref{Poisson-Coh-lem1}, we have given a description of the operator $d_\pi$ in the complex
\begin{equation*}
H^{0}(X,\mathcal{O}_{X})\xrightarrow {d_{\pi}}
H^{0}(X,T_{X})\xrightarrow{d_{\pi}}H^{0}(X,\wedge^{2}T_{X})
\xrightarrow{d_{\pi}}
\ldots \xrightarrow{d_{\pi}}H^{0}(X,\wedge^{n}T_{X}).
\end{equation*}

\begin{lemma}\label{Poisson-Coh-lem2}
Let $\pi$ be a holomorphic toric Poisson structure on $X=\CP^n$.
We have
\begin{equation}
d_\pi(V_I^k)\subseteq V_I^{k+1}
\end{equation}
for all $I\in S_{k}$ $(0\leq k\leq n)$,  where
\begin{enumerate} 
\item $V_I^k$ is considered 
as a subspace of $H^{0}(X,\wedge^{k}T_{X})$;
\item $d_\pi(V_I^k)$ denotes the image of 
$V_I^k$ under the map 
$H^{0}(X,\wedge^k T_{X})\xrightarrow{d_{\pi}}H^{0}(X,\wedge^{k+1}T_{X})$; 
\item $V_I^{k+1}$ is considered as a subspace of 
$H^{0}(X,\wedge^{k+1}T_{X})$.
\end{enumerate}
\end{lemma}

\begin{proof}
Let $\Psi=\chi^I\cdot \cV_I\wedge w$ be a holomorphic $k$-vector field in 
$V_I^k$, where $I\in S_k$ and $w\in W^{k-|I|}$.
By Lemma \ref{Poisson-Coh-lem1}, we have 
\begin{gather*}
d_\pi(\Psi)=[\pi,\chi^I\cdot \cV_I\wedge w]
=\chi^I\cdot\rho(\imath_{I}\Pi)\wedge(\cV_I\wedge w)\\
=(-1)^{|I|}\chi^I\cdot \cV_I\wedge (\rho(\imath_{I}\Pi)\wedge w).
\end{gather*}
As $\rho(\imath_{I}\Pi)\wedge w\in W^{k-|I|+1}$, $d_\pi(\Psi)$ is an element in 
$V_I^{k+1}$.

Thus 
$$d_\pi(V_I^k)\subseteq V_I^{k+1}
\quad \text{ for all} \quad I\in S_k.$$
\end{proof}

By Lemma \ref{Poisson-Coh-lem2}, we have the chain complex
\begin{equation}
\C (\chi^I\cdot \cV_I)\xrightarrow{d_\pi}\C (\chi^I\cdot \cV_I)\wedge W\xrightarrow{d_{\pi}}\C (\chi^I\cdot \cV_I)\wedge W^2\xrightarrow{d_\pi}\cdots\xrightarrow{d_\pi}\C (\chi^I\cdot \cV_I)\wedge W^{n-|I|}
\end{equation}
for all $I\in M$, where $\C (\chi^I\cdot \cV_I)\wedge W^{k}$ is considered as a subspace
of $H^0(X,\wedge^{|I|+k}TX)$ for $0\leq k\leq n-|I|$.

\begin{lemma}\label{Poisson-Coh-lem3}
Let $\pi$ be a holomorphic toric Poisson structure on $X=\CP^n$.
\begin{enumerate}
\item
 Given any $I\in S_k(\pi)$ $(0\leq k\leq n)$, \i.e.,
$I\in S_k$ satisfying the equation
$$(\imath_{I}\Pi)\wedge \cE_I=0,$$ 
we have
\begin{equation}
d_\pi(V_I^k)=0.
\end{equation}
\item
Given any $I\in S_{k-1}\subset S_k$ $(1\leq k\leq n)$ satisfying the equation
$$(\imath_{I}\Pi)\wedge \cE_I\neq0,$$ 
if $\Psi$ is a holomorphic $k$-vector field in $V_I^k$ satisfying 
$$d_\pi(\Psi)=0,$$
 then there exists a holomorphic $(k-1)$-vector field $\Phi$
 in $V_I^{k-1}$, such that
 $$\Psi=d_\pi(\Phi).$$  
 \end{enumerate}
\end{lemma}

\begin{proof}
\begin{enumerate}
\item For any $I\in S_k$ satisfying the equation
$$(\imath_{I}\Pi)\wedge \cE_I=0,$$
we have
$$\rho((\imath_{I}\Pi)\wedge \cE_I)=\rho(\imath_{I}\Pi)\wedge \cV_I=0.$$
By Equation \eqref{PoCoh-lemma-eqn2.5} in Lemma \ref{Poisson-Coh-lem1}, 
we have  
$$d_\pi(V_I^k)=0.$$

\item Suppose that $\Psi=\chi^I\cdot \cV_I\wedge w$, 
where $I\in S_{k-1}\subseteq S_k$ and $w\in W^{k-|I|}$.

By Lemma \ref{Poisson-Coh-lem1}, we have
\begin{gather*}
d_{\pi}(\Psi)=[\pi,\chi^I\cdot \cV_I\wedge w]
=\rho(\imath_{I}\Pi)\wedge(\chi^I\cdot \cV_I\wedge w)\\
=\chi^I\cdot\rho(\imath_{I}\Pi)\wedge v_{i_1}\wedge\ldots\wedge v_{i_l}\wedge w.
\end{gather*}

The condition $d_{\pi}(\Psi)=0$ implies 
$$\rho(\imath_{I}\Pi)\wedge v_{i_1}\wedge\ldots\wedge v_{i_l}\wedge w=0.$$
And the condition $(\imath_{I}\Pi)\wedge \cE_I\neq0$ implies that
\begin{itemize}
\item $\imath_{I}\Pi$ and $e_{i_1},e_{i_2}\ldots e_{i_l}$ are
$\C$-linear independent vectors in $N_\C$; 
\item $\rho(\imath_{I}\Pi)$ and $v_{i_1},v_{i_2}\ldots v_{i_l}$ are
$\C$-linear independent vectors in $W=\rho(N_\C)$. 
\end{itemize}

By simple linear algebra we know that $w\in W^{k-|I|}$ can be written as
$$w=\rho(\imath_{I}\Pi)\wedge w_0+\sum_{s=1}^{l}v_{i_s}\wedge w_i,$$  
where $w_0,w_1,\ldots,w_l$ are elements in $W^{k-|I|-1}$.
As a consequence, we have
\begin{align*} 
\Psi&=\chi^I\cdot \cV_I\wedge w\\
&=\chi^I(v_{i_1}\wedge v_{i_2}\wedge\ldots\wedge v_{i_l})\wedge
(\rho(\imath_{I}\Pi)\wedge w_0+\sum_{s=1}^{l}v_{i_s}\wedge w_s)\\
&=\chi^I(v_{i_1}\wedge v_{i_2}\wedge\ldots\wedge v_{i_l})\wedge
(\rho(\imath_{I}\Pi)\wedge w_0)\\
&=(-1)^{|I|}\rho(\imath_{I}\Pi)\wedge(\chi^I\cdot \cV_I\wedge w_0)
\end{align*}

Let $\Phi=(-1)^{|I|}\chi^I\cdot \cV_I\wedge w_0$.
Then $\Phi$ is a holomorphic $(k-1)$-vector field in the space 
$V_I^{k-1}$. 
And by Lemma \ref{Poisson-Coh-lem1}, we have
$$\Psi=d_{\pi}(\Phi).$$
\end{enumerate}
\end{proof}

{\bf Proof of Theorem \ref{CPn-cohomology-thm}:}
\begin{proof}
\begin{enumerate}
\item In the case of $k=0$, $S_0$ consists of only $(0,...0)\in M$, 
and $W^{k-|I|}=W^0=\C$. 
Hence we have $$H_\pi^0(X)=\C.$$

\item In the case of $1\leq k\leq n$, by Theorem \ref{CPn-multiVec-thm}, 
any holomorphic $k$-vector field $\Psi\in H^{0}(X,\wedge^{k}T_{X})$ can be written as
$$\Psi=\sum_{I\in S_k}\Psi_I,\qquad \Psi_I \in V_I^k.$$ 
As $$d_\pi(\Psi)=\sum_{I\in S_k}d_\pi(\Psi_I),$$
by Lemma \ref{Poisson-Coh-lem2} and Theorem \ref{CPn-multiVec-thm}, we have that
$$d_\pi(\Psi)=0\Longleftrightarrow d_\pi(\Psi_I)=0\quad\text{for all}\quad I\in S_k.$$
\begin{enumerate}
\item For any $I\in S_{k-1}\subseteq S_k$ satisfying 
$(\imath_{I}\Pi)\wedge \cE_I\neq0$, by Lemma \ref{Poisson-Coh-lem3}, 
$d_\pi(\Psi_I)=0$ imples that there exist 
$\Phi_I\in V_I^{k-1}$ such that
$\Psi_I=d_\pi(\Phi_I).$ Thus there exists only zero Poisson cohomology class in 
$V_I^k$ for $I\in S_{k-1}\subseteq S_k$ satisfying 
$(\imath_{I}\Pi)\wedge \cE_I\neq0$.
\item For any $I\in S_{k-1}\subseteq S_k$ satisfying 
$(\imath_{I}\Pi)\wedge \cE_I=0$, i.e., $I\in S_{k-1}(\pi)$,  
by Lemma \ref{Poisson-Coh-lem3}, we have that
$$d_\pi(V_I^{k-1})=0
\quad\text{and}\quad d_\pi(V_I^k)=0.$$
By Lemma \ref{Poisson-Coh-lem3} and Theorem \ref{CPn-multiVec-thm},
any nonzero element $\Psi_I\in V_I^k$
satisfying $I\in S_{k-1}(\pi)$ 
represents a nonzero cohomology class in the Poisson cohomology group.
Moreover, the space
$$\bigoplus_{I\in S_{k-1}(\pi)}V_I^k$$
can be considered as a subgroup of the Poisson cohomology group $H_{\pi}^k(X)$.
\item For any $I\in S_k(\pi)\backslash S_{k-1}(\pi)$, i.e, $I\in S(k)$ satisfying
$(\imath_{I}\Pi)\wedge \cE_I=0$, by Lemma \ref{Poisson-Coh-lem3}, we have
$$d_\pi(V_I^k)=0.$$
By Theorem \ref{CPn-multiVec-thm}, we have
$$H^{0}(X,\wedge^{k-1}T_{X})=\bigoplus_{I\in S_{k-1}}V_I^{k-1}.$$
And by Lemma \ref{Poisson-Coh-lem2}, we have 
$$d_\pi(H^{0}(X,\wedge^{k-1}T_{X}))\subseteq
\bigoplus_{I\in S_{k-1}}V_I^k.$$
Thus any nonzero element $\Psi_I\in V_I^k$
satisfying $I\in S_{k}(\pi)\backslash S_{k-1}(\pi)$ represents a nonzero 
Poisson cohomology class.
Moreover, the space
 $$\bigoplus_{I\in S_k(\pi)\backslash S_{k-1}(\pi)}V_I^k$$
can be considered as a subgroup of the Poisson cohomology group $H_{\pi}^k(X)$.
\end{enumerate}
By the arguments above, we have 
$$H_{\pi}^k(X)=\bigoplus_{I\in S_k(\pi)}V_I^k$$
for $1\leq k\leq n$.

\item In the case of $k>n$, $H_{\pi}^k(X)=0$ comes directly from 
Lemma \ref{LSX-lem} and Theorem \ref{TheoemVanish}.
\end{enumerate}
\end{proof}

\subsection{Poisson cohomology of the standard Poisson structure on $\CP^n$}

By Proposition \ref{STPoisson-prop}, the standard Poisson structure on 
$X=\CP^n$ can be written as
$$\pi_{st}=\sum_{1\leq i<j\leq n}v_i\wedge v_j,$$
where $v_i=p_*(z_i\frac{\partial}{\partial z_i})=\rho(e_i)$ $(0\leq i\leq n)$.
And we have $\Pi_{st}=\sum_{1\leq i<j\leq n} e_i\wedge e_j\in\wedge^2 N_\C$.

We can apply Theorem \ref{CPn-cohomology-thm} to compute the Poisson cohomology 
of the standard Poisson structure on $\CP^n$.
Here we only list the Poisson cohomology groups in the case $n=2$ and $n=3$. 
For other cases it could be done similarly, but the computation will be more complicated. 

\begin{proposition}\label{CP2-coh-prop}
The standard Poisson structure on $X=\CP^2$ can be written as
$$\pi_{st}= v_1\wedge v_2.$$
The Poisson cohomology groups of $(X, \pi_{st})$ are as following:
\begin{enumerate}
\item $H_{\pi_{st}}^0(X)=\C$, $\dim H_{\pi_{st}}^0(X)=1$.
\item $H_{\pi_{st}}^1(X)$ has a basis $\{v_1, v_2\}$, $\dim H_{\pi_{st}}^1(X)=2$.
\item $H_{\pi_{st}}^2(X)$ has a basis 
$\{(z_0^{m_0}z_1^{m_1}z_2^{m_2} )v_1\wedge v_2\}$ with  
$(m_0,m_1,m_2)$ in the set
\[
\left\{
\begin{array}{ccc}
(0,0,0), &  & \\
(-1,-1,2), &(-1,2,-1), & (2,-1,-1)
\end{array}
\right\}.
\]
Hence $\dim H_{\pi_{st}}^2(X)=4$.
\item $H_{\pi_{st}}^k(X)=0$ for $k>2$.
\end{enumerate}
\end{proposition}
The Proposition \ref{CP2-coh-prop} verified the results about Poisson cohomology of $\CP^2$ in \cite{Hong-Xu 11}.

\begin{proposition}\label{CP3-coh-prop}
The standard Poisson structure on $X=\CP^3$ can be written as
$$\pi_{st}= v_1\wedge v_2+v_1\wedge v_3+v_2\wedge v_3.$$
The Poisson cohomology groups of $(X, \pi_{st})$ are as following:
\begin{enumerate}
\item $H_{\pi_{st}}^0(X)=\C$, $\dim H_{\pi_{st}}^0(X)=1$.
\item $H_{\pi_{st}}^1(X)$ has a basis $\{v_1, v_2, v_3\}$, $\dim H_{\pi_{st}}^1(X)=3$.
\item $H_{\pi_{st}}^2(X)$ has a basis, which is the union of the three parts:
\begin{enumerate}
\item $\{v_1\wedge v_2,\quad v_1\wedge v_3, \quad v_2\wedge v_3\}$,
\item $\{(z_0^{m_0}z_1^{m_1}z_2^{m_2} z_3^{m_3})v_0\wedge v_2\}$ with
$(m_0,m_1,m_2,m_3)$ in the set 
\[
\left\{
\begin{array}{ccc}
(-1, 1,,-1,1),& (-1,2,-1,0),& (-1,0,-1,2)
\end{array}
\right\}
\]
and $v_0=-\sum_{i=1}^3 v_i$,
\item $\{(z_0^{m_0}z_1^{m_1}z_2^{m_2} z_3^{m_3})v_1\wedge v_3\}$ with 
$(m_0,m_1,m_2,m_3)$ in the set 
\[
\left\{
\begin{array}{ccc}
(1, -1,,1,-1),& (2,-1,0,-1),& (0,-1,2,-1)
\end{array}
\right\}.
\]
\end{enumerate}
Hence $\dim H_{\pi_{st}}^2(X)=9$.
\item $H_{\pi_{st}}^3(X)$ has a basis
$\{(z_0^{m_0}z_1^{m_1}z_2^{m_2}z_3^{m_3} )v_1\wedge v_2\wedge v_3\}$ with  
$(m_0,m_1,m_2,m_3)$ in the set
\[\left\{
\begin{array} {cccc}
(0,0,0,0), & & & \\
(-1, 1,,-1,1),  & (-1,2,-1,0), & (-1,0,-1,2), &\\
(1, -1,1,-1), &(2,-1,0,-1),  &(0,-1,2,-1),&  \\
(-1,-1,-1,3), & (-1,-1,3,-1), & (-1, 3,-1,-1), & (3,-1,-1,-1)
\end{array}
\right\}.\]
Hence $\dim H_{\pi_{st}}^3(X)=11$. 
\item $H_{\pi_{st}}^k(X)=0$ for $k>3$.
\end{enumerate}
\end{proposition}

For general $\CP^n$ equipped with the standard Poisson structure, we can compute the Poisson cohomology groups by Theorem \ref{CPn-cohomology-thm}. Here we have a result for the first Poisson cohomology group of $(\CP^n, \pi_{st})$.

\begin{theorem}\label{ST-thm}
For $X=\CP^{n}$ equipped with the standard Poisson structure
$$\pi_{st}=\sum_{1\leq i<j\leq n}v_i\wedge v_j,$$
we have 
\begin{equation*}
 H_{\pi_{st}}^1(X)=W \quad \text{and}\quad \dim H_{\pi_{st}}^1(X)=n.
 \end{equation*}
\end{theorem}

To prove Theorem \ref{ST-thm}, we need the following lemma. 
 
\begin{lemma}\label{ST-thm-lem}
For $X=\CP^{n}$ equipped with the standard Poisson structure
$$\pi_{st}=\sum_{1\leq i<j\leq n}v_i\wedge v_j,$$
we have 
$$S(1,\pi_{st})=\emptyset,$$
where $S(1,\pi_{st})$ is the set of all $I\in S(1)$ satisfying the Equation
\begin{equation}\label{ST-thm-proof-eqn1}
(\imath_{I}\Pi_{st})\wedge \cE_I=0,
\end{equation}
here $\Pi_{st}=\sum_{1\leq i<j\leq n}e_i\wedge e_j$.
\end{lemma}
\begin{proof}

Let us denote $\alpha_{i,j}$ $(0\leq i\neq j\leq n)$ as the element in $M$, with
$m_i=-1$, $m_j=1$ and $0$ elsewhere. Then we have 
$$S(1)=\{\alpha_{i,j}\mid 0\leq i\neq j\leq n\}.$$
In the case of $I=\alpha_{i,j}$, we have $\cE_I=e_i$. The Equation \eqref{ST-thm-proof-eqn1} becomes
\begin{equation}\label{ST-thm-proof-eqn2}
(\imath_{I}\Pi_{st})\wedge e_i=0,
\end{equation}
which is equivalent to
\begin{equation}\label{ST-thm-proof-eqn3}
\imath_{I}\Pi_{st}=\lambda e_i,  \qquad \lambda\in\C.
\end{equation}
Suppose that $\{e_1^*,e_2^*,...,e_n^*\}\subset M$ is the dual basis of 
$\{e_1,e_2,...,e_n\}$. Then we have
\begin{align*}
I=
\begin{cases}
e_j^*, \quad & \text{for} \quad i=0,~j\neq 0.\\
-e_i^*, \quad &\text{for}\quad i\neq 0,~ j=0.\\
-e_i^*+e_j^*,\quad &\text{for}\quad i\neq 0,~ j\neq 0.
\end{cases}
\end{align*} 
\begin{itemize}
\item[(a)] In the case $i=0$ and $j\neq 0$, Equation \eqref{ST-thm-proof-eqn3} becomes
$$\imath_{e_j^*}\Pi=\lambda e_0=-\lambda\sum_{k=1}^{n}e_k,$$
which implies
\begin{equation}\label{ST-thm-proof-eqn4}
\Pi(e_j^*,e_k^*)=\langle \imath_{e_j^*}\Pi,e_k^*\rangle
=\langle-\lambda\sum_{k=1}^{n}e_k, e_k^*\rangle=-\lambda
\end{equation}
for all $1\leq k\leq n$. 

But $$\Pi_{st}=\sum_{1\leq s<t\leq n}e_s\wedge e_t,$$
which imples
$$\Pi(e_j^*,e_j^*)=0$$
and $$ \Pi(e_j^*,e_k^*)=\pm1$$
for all $1\leq k\neq j\leq n$. Thus Equation \eqref{ST-thm-proof-eqn4} can not hold for all $1\leq k\leq n$.

So Equation \eqref{ST-thm-proof-eqn3} has no solution in this case.

\item[(b)] In the case $i\neq 0$ and $ j=0$, Equation \eqref{ST-thm-proof-eqn3} becomes
$$\imath_{(-e_i^*)}\Pi=\lambda e_i,$$
which implies
$$ \Pi(-e_i^*,e_k^*)=0$$
for all $1\leq k\neq i\leq n$.

But $$\Pi_{st}=\sum_{1\leq s<t\leq n}e_s\wedge e_t,$$ which implies
$$ \Pi(-e_i^*,e_k^*)=\pm 1$$
for all $1\leq k\neq i\leq n$. 

Thus Equation \eqref{ST-thm-proof-eqn3} has no solution in this case.

\item[(c)] In the case $ i\neq 0$ and $j\neq 0$, $I=-e_i^*+e_j^*$, 
Equation \eqref{ST-thm-proof-eqn3} becomes
$$\imath_{(-e_i^*+e_j^*)}\Pi=\lambda e_i.$$
It can not be true since that 
$$ \Pi(-e_i^*+e_j^*, e_j^*)=\Pi(-e_i^*, e_j^*)=\pm 1,$$
but $$ \langle\lambda e_i,e_j^*\rangle=0.$$
\end{itemize}
By the arguments above, we have $S(1,\pi_{st})=\emptyset.$
\end{proof}

{\bf Proof of Theorem \ref{ST-thm}:}
\begin{proof}
By Theorem \ref{CPn-cohomology-thm}, we have
$$H_{\pi_{st}}^1(X)=\bigoplus_{I\in S_1(\pi_{st})}\C (\chi^I\cdot \cV_I)\wedge W^{1-|I|}.$$
By Lemma \ref{ST-thm-lem}, $S(1,\pi_{st})=\emptyset$,
we have 
$S_1(\pi_{st})=S(0,\pi_{st})\cup S(1,\pi_{st})=S(0,\pi_{st})$.

As $S(0,\pi_{st})$ consists only one element $I=(0,\ldots,0)$, we have 
$$\bigoplus_{I\in S(0,\pi_{st})}\C (\chi^I\cdot \cV_I)\wedge W^{1-|I|}=W.$$

Therefore we have
$$ H_{\pi_{st}}^1(X)=W \quad \text{and} \quad \dim H_{\pi_{st}}^1(X)=n.$$
 \end{proof}

Let $\sigma: \CP^n\rightarrow\CP^n$ be an isomorphism of $\CP^n$ defined by
\begin{equation} \label{ST-CPn-Zn-eqn}
\sigma([z_0, z_1, \ldots, z_{n-1}, z_n])=[z_1, z_2, \ldots, z_n, z_0].
\end{equation}
 
As $\sum_{i=0}^{n}v_i=0$, we have that
\begin{align*}
\pi_{st}=&\sum_{1\leq i<j\leq n}v_i\wedge v_j=\sum_{0\leq i<j\leq n}v_i\wedge v_j\\
=&\sum_{0\leq i<j\leq n-1}v_i\wedge v_j\\
=&\sigma^{-1}_{*}(\pi_{st}),
\end{align*}
where the last step holds since that 
$$\sigma_*(v_i)=\sigma_*(z_i\frac{\partial}{\partial z_i})
=z_{i+1}\frac{\partial}{\partial z_{i+1}}=v_{i+1}\quad(0\leq i\leq n-1).$$

Therefore the standard Poisson structure is invariant under the map $\sigma$.

Since $\sigma^{n+1}=\id$, we can define a $\Z_{n+1}$-action on $X=\CP^n$, 
generated by $\sigma$. By the arguments above, 
$\pi_{st}$ is invariant under the $\Z_{n+1}$-action.

\begin{corollary}\label{ST-prop-invariant}
The standard Poisson structure
$$\pi_{st}=\sum_{1\leq i<j\leq n}v_i\wedge v_j$$
on $X=\CP^n$ is invariant under the $\Z_{n+1}$-action defined above,
which induces a  $\Z_{n+1}$-action on the Poisson cohomology groups $H^{k}_{\pi_{st}}(X)$ $(0\leq k\leq n)$.
\end{corollary}

In the cases of $X=\CP^2$ (Proposition \ref{CP2-coh-prop}) 
and $X=\CP^3$ (Proposition \ref{CP3-coh-prop}), it is easy to verify the $\Z_3$-action 
and the $\Z_4$-action on the Poisson cohomology groups.

\section{Poisson cohomology of holomorphic toric Poisson structures on $\C^n$}\label{section5}

In this section, we will study the Poisson cohomology of $X=\C^n$.
Our setting is based on Example \ref{Cn-exa1} and Example \ref{Cn-exa2}.

Let $W=\rho(N_\C)$, $W^k=\wedge^k W$ and $W^0=\C$.
For any $I\in M$, let
$$m_i=\langle I, e_i\rangle~(1\leq i\leq n).$$
Suppose that $\{m_{i_1},m_{i_2},...,m_{i_l}\mid 0\leq i_1<i_2<...<i_l\leq n\}$ is the set of all the elements equal to $-1$ in $\{m_0, m_1,\ldots, m_n\}$. 
Set
\begin{align*}
|I|=l,\quad \chi^I=z_1^{m_1}\cdots z_n^{m_n},\quad \cV_I=v_{i_1}\wedge\ldots\wedge v_{i_l}\in W^l, \quad
\cE_I=e_{i_1} \wedge\ldots\wedge e_{i_l} \in\wedge^l N.
\end{align*}

Any holomorphic toric Poisson structure $\pi$ on $X=\C^n$ can be written as 
$$\pi=\sum_{1\leq i<j\leq n}a_{ij}v_i\wedge v_j.$$
The corresponding $\Pi=\sum_{1\leq i<j\leq n}a_{ij}e_i\wedge e_j\in\wedge^2 N_{\C}$
satisfies $\rho(\Pi)=\pi$.
Notice that  
$$\pi=\sum_{1\leq i<j\leq n}a_{ij}v_i\wedge v_j=\sum_{1\leq i<j\leq n}a_{ij}z_i z_j\frac{\partial}{\partial z_i}\wedge\frac{\partial}{\partial z_j},$$ 
which is called a diagonal quadratic Poisson structure in \cite{Monnier 02}.

For $X=\C^n$, since $H^i(X,\wedge^j TX)=0$ $(i>0)$, the Poisson cohomology  
$H^\bullet_\pi (X)$ is isomorphic to the cohomology of the complex
\begin{equation*}
H^{0}(X,\mathcal{O}_{X})\xrightarrow {d_{\pi}}
H^{0}(X,T_{X})\xrightarrow{d_{\pi}}
H^{0}(X,\wedge^{2}T_{X})\xrightarrow{d_{\pi}}
\ldots \xrightarrow{d_{\pi}}H^{0}(X,\wedge^{n}T_{X}),
\end{equation*}
where $d_{\pi}=[\pi,\cdot]$. 

The complex vector space $H^0(X,\wedge^k TX)$ $(0\leq k\leq n)$ is infinite dimensional. Any holomorphic $k$-vector field on $X$ can be written as  
\begin{equation*}
\sum_{1\le i_1<i_2<\ldots< i_k\le n}f_{i_1,i_2,\ldots,i_k}\frac{\partial}{\partial z_{i_1}}\wedge\ldots\wedge\frac{\partial}{\partial z_{i_k}},
\end{equation*}
where $f_{i_1,i_2,\ldots,i_k}$ are analytic functions on $\C^n$ with variables $z_1,z_2,\ldots, z_n$. 
Analytic functions on $\C^n$ can be written as 
$$\sum_{m_i\geq 0} a_{m_1,\ldots, m_n}z_1^{m_1}\ldots z_n^{m_n}$$
satisfying the convergence conditions, which makes the computation of Poisson cohomology difficult.

However, we can consider the algebraic Poisson coholomogy groups and the formal Poisson cohomology groups, which are easier to compute than the Poisson cohomology of holomorphic toric Poisson structures on $\C^n$. 
\begin{itemize}
\item In \cite{Monnier 02}, Monnier considered the formal $k$-vector fields on $R^n$ and computed the formal Poisson cohomology of diagonalizable quadratic Poisson structures. His results also applies for $X=\C^n$. 
A formal  $k$-vector field on $\C^n$ is in the following form
\begin{equation*}
\sum_{1\le i_1<i_2<\ldots< i_k\le n}f_{i_1,i_2,\ldots,i_k}\frac{\partial}{\partial z_{i_1}}\wedge\ldots\wedge\frac{\partial}{\partial z_{i_k}},
\end{equation*}
where $f_{i_1,i_2,\ldots,i_k}\in\C[[z_1,z_2,\ldots, z_n]]$. 
We denote $\F H^0(X,\wedge^k TX)$ $(0\leq k\leq n)$ as the complex vector space of all 
formal  $k$-vector fields on $X=\C^n$. Then $H^0(X,\wedge^k TX)$ can be considered as a subspace of $\F H^0(X,\wedge^k TX)$. 

The formal Poisson cohomology group $\F H_{\pi}^{\bullet}(X)$ is defined to be the cohomology of the following complex
\begin{equation*}
\F H^{0}(X,\mathcal{O}_{X})\xrightarrow {d_{\pi}}
\F H^{0}(X,T_{X})\xrightarrow{d_{\pi}}
\F H^{0}(X,\wedge^{2}T_{X})\xrightarrow{d_{\pi}}
\ldots \xrightarrow{d_{\pi}}\F H^{0}(X,\wedge^{n}T_{X}),
\end{equation*}
where $d_{\pi}=[\pi,\cdot]$. 
The formal Poisson cohomology group $\mathcal{F} H_{\pi}^{\bullet}(X)$ can also be considered as the Poisson cohomology group of the Poisson algebra $\C[[z_1,z_2,\ldots, z_n]]$ on $X=\C^{n}$, with the Poisson bracket defined by $\pi$.

\item In \cite{Aaron 13}, McMillan studied the Poisson algebra of the coordinate ring of an affine Poisson variety. Here we only focus on the case of $X=\C^n$. For an algebraic $k$-vector field on $\C^n$, we means a $k$-vector field with the following form
\begin{equation*}
\sum_{1\le i_1<i_2<\ldots< i_k\le n}f_{i_1,i_2,\ldots,i_k}\frac{\partial}{\partial z_{i_1}}\wedge\ldots\wedge\frac{\partial}{\partial z_{i_k}},
\end{equation*}
where $f_{i_1,i_2,\ldots,i_k}\in\C[z_1,z_2,\ldots, z_n]$. 
We denote $\mathcal{A} H^0(X,\wedge^k TX)$ $(0\leq k\leq n)$ as the complex vector space of all algebraic  $k$-vector fields on $X=\C^n$. 
Then $\A H^0(X,\wedge^k TX)$ can be considered as a subspace of 
$H^0(X,\wedge^k TX)$. 

The algebraic Poisson cohomology group $\mathcal{A} H_{\pi}^{\bullet}(X)$ is defined to be the cohomology of the following complex
\begin{equation*}
\A H^{0}(X,\mathcal{O}_{X})\xrightarrow {d_{\pi}}
\A H^{0}(X,T_{X})\xrightarrow{d_{\pi}}
\A H^{0}(X,\wedge^{2}T_{X})\xrightarrow{d_{\pi}}
\ldots \xrightarrow{d_{\pi}}\A H^{0}(X,\wedge^{n}T_{X}),
\end{equation*}
where $d_{\pi}=[\pi,\cdot]$. 
The algebraic Poisson cohomology group $\mathcal{A} H_{\pi}^{\bullet}(X)$ can also be considered as the Poisson cohomology group of the Poisson algebra $\C[z_1,z_2,\ldots, z_n]$ on $X=\C^{n}$, with the Poisson bracket defined by $\pi$.
\end{itemize}

\begin{lemma}\label{Cn-multiVec-lem}
Let $X=\C^n$. Let $V_I^k=\C (\chi^I\cdot \cV_I)\wedge W^{k-|I|}$ for $|I|\leq k\leq n$.
Then we have
\begin{enumerate}
\item The space of formal $k$-vector fields
\begin{equation}
\F H^0(X,\wedge^k \mathcal{T}_X)=\prod_{I\in S_k}V_I^k
\end{equation}
for $0\leq k\leq n$, where $S_k$ is the set consisting of all $I\in M$ satisfying the conditions
\begin{equation*} 
m_i\geq -1 \quad (1\leq i\leq n)
\end{equation*}
and 
\begin{equation*}
|I|\leq k.
\end{equation*}

\item The space of algebraic $k$-vector fields
\begin{equation}
\mathcal{A} H^0(X,\wedge^k \mathcal{T}_X)=\bigoplus_{I\in S_k}V_I^k
\end{equation}
for $0\leq k\leq n$, where $S_k$ is the set consisting of all $I\in M$ satisfying the conditions
\begin{equation*} 
m_i\geq -1 \quad (1\leq i\leq n)
\end{equation*}
and 
\begin{equation*}
|I|\leq k.
\end{equation*}

\end{enumerate}
\end{lemma}

\begin{theorem}\label{Cn-cohomology-thm}
Let $\pi$ be a holomorphic toric Poisson structure on $X=\C^n$. 
Let $V_I^k=\C (\chi^I\cdot \cV_I)\wedge W^{k-|I|}$ for $|I|\leq k\leq n$.
Then we have
\begin{enumerate}
\item for $0\leq k\leq n$, we have
\begin{align}
&\F H_{\pi}^k(X)=\prod_{I\in S_k(\pi)}V_I^k,\\
&\mathcal{A} H_{\pi}^k(X)=\bigoplus_{I\in S_k(\pi)}V_I^k,
\end{align}
where $S_k(\pi)$ is the set consisting of all $I\in M$ satisfying 
\begin{align}
&m_i\geq -1\quad (0\leq i\leq n); \label{Cn-cohomology-thm-Con1}\\  
&|I|\leq k; \label{Cn-cohomology-thm-Con2}
\end{align}
and the equation
\begin{equation}
(\imath_{I}\Pi)\wedge \cE_I=0.\label{Cn-cohomology-thm-Con3}
\end{equation}
\item  $\F H_{\pi}^k(X)=\mathcal{A} H_{\pi}^k(X)=0$ for $k>n$.
\end{enumerate}
\end{theorem}

Lemma \ref{Cn-multiVec-lem} and Theorem \ref{Cn-cohomology-thm} can be proved by similar arguments as we have done for $\CP^n$. Here we skip the proof.

\begin{remark}
In Lemma \ref{Cn-multiVec-lem} and Theorem \ref{Cn-cohomology-thm}, the results for formal $k$-vector fields and formal Poisson cohomology were proved by Monnier in \cite{Monnier 02}.  Here we write the results from the viewpoint of toric varieties.
\end{remark}

By Theorem \ref{Cn-cohomology-thm}, we have the following corollary.
\begin{corollary}
If $S_k(\pi)$ is a finite set, then we have $\F H_{\pi}^k(X)\cong \mathcal{A}H_{\pi}^k(X)$.
\end{corollary}

{\bf Question:} if $S_k(\pi)$ is a finite set, do we have 
$\F H_{\pi}^k(X)\cong \mathcal{A}H_{\pi}^k(X)\cong H_{\pi}^k(X)$?

\end{document}